\documentclass[a4paper,12pt,DIV=calc]{scrartcl}
\usepackage[utf8]{inputenc}
\usepackage{amsmath,amsfonts,amssymb,amsthm,mathtools}
\usepackage{amsmath}
\usepackage{amsfonts}
\usepackage{amssymb}
\usepackage{dsfont}
\usepackage{mathrsfs}
\usepackage[T1]{fontenc}
\usepackage[ngerman,english]{babel}
\usepackage[inline]{enumitem}
\usepackage{latexsym,graphicx}
\usepackage[all]{xy}
\usepackage{mathtools}
\usepackage{color}
\usepackage{authblk}
\usepackage{amscd}
\usepackage{microtype}

\numberwithin{equation}{section}

\newtheorem{theorem}{Theorem}

\newtheorem{lemma}{Lemma}

\newtheorem{remark}{Remark}

\newtheorem{assumption}{Assumptions}

\begin{document}

\title{A lower bound for the BCS functional with boundary conditions at infinity}

\author{Andreas Deuchert}

\date{}

\maketitle

\begin{abstract} 
We consider a many-body system of fermionic atoms interacting via a local pair potential and subject to an external potential within the framework of BCS theory. We measure the free energy of the whole sample with respect to the free energy of a reference state which allows us to define a BCS functional with boundary conditions at infinity. Our main result is a lower bound for this energy functional in terms of expressions that typically appear in Ginzburg-Landau functionals.
\end{abstract}


\section{Introduction, set-up and main results}
\label{sec:introduction_and_main_results}
\subsection{Introduction}
\label{subsec:introduction}
After the phenomenon of superconductivity was discovered in 1911 by Heike Kamerlingh Onnes, it took almost 46 years of intensive research before the microscopic origin of this remarkable effect could be explained. It was in 1957 when Bardeen, Cooper and Schrieffer (BCS) published their famous paper with the title "Theory of Superconductivity" \cite{BCS}, in which they introduced the first, generally accepted, microscopic model of superconductivity based on the idea of electron pairing due to an effective attraction mediated by phonons. In recognition of this work they were awarded the Nobel prize in 1972. \newline
After this succuess in understanding the phenomenon of superconductivity \cite{Parks}, it was realized that the mathematical ideas underlying the BCS theory of superconductivity and the physical idea of pairing in fermionic systems can also be used to describe the normal to superfluid phase transition in liquid helium-3 and in cold fermionic gases \cite{Leggett}. In this situation, one has to replace the usual non-local phonon-induced interaction in the gap equation by a local pair potential suitable to describe typical interactions occuring in these systems. Due to numerous breakthroughs in the experimental realisation of bosonic and fermionic quantum gases in the past 20 years, this direction of research has gained increasing importance in physics as well as in mathematics \cite{coldatoms1,coldatoms2,reviewBCS}. \newline
Apart from being a paradigmatic model in solid state physics and in the theory of cold quantum gases, the BCS theory of superconductivity/superfluidity, that is the BCS gap equation and the BCS functional, shows a rich mathematical structure which has been well recognised. See \cite{Odeh,Billard_Fano,Vansevenant,Yang1, McLeod_Yang,Yang2} for works on the gap equation with interaction kernels suitable to describe the physics of conduction electrons in solids and \cite{BCS_general_pair_interaction,Critical_temperature_BCS,Critical_temperature_and_energy_gap,Spectral_properties_of_the_BCS_Eq,
Gerhard1,spin_polarized_BCS,Frank_Lemm,symmetrybreaking} for works that treat the translation-invariant BCS functional with a local pair interaction.  \newline
Recently, there has also been considerable interest in the BCS model with external electric and/or magnetic fields. While in the translation-invariant setting, the BCS gap equation is a nonlinear integral equation for a function depending on one space variable, this is no longer true if external fields are added to the model. On the level of the gap equation one is faced with a nonlinear operator equation while in the case of the BCS functional it is a non-commutative variational problem. To our knowledge, the first mathematical work in which the BCS functional with external fields has been considered is due to Frank, Hainzl, Seiringer and Solovej who managed to rigorously prove a connection between the microscopic BCS theory and the equally famous macroscopic Ginzburg-Landau theory of superconductivity \cite{GL_original}. In a non-rigorous way this connection has been discovered by Gorkov in 1959 \cite{Gorkov}. They show, in a certain scaling limit where the temperature of the sample is assumed to be close to the critical temperature of the translation-invariant model, that the BCS free energy equals the free energy of the normal state plus a correction, which to leading order is given by the infimum of a suitably chosen Ginzburg-Landau functional. Using similar techniques, the same authors investigate in \cite{external_field_T_c} the influence of the external fields on the critical temperature. Other works considering BCS theory with external fields from various perspectives are \cite{linearTcdependence,BCS_low_density1,BCS_low_density2,Frank_Lemm2,SigalBdG}. In \cite{Chris_Schlein,dynamics_1,dynamics_2} the equations are investigated from the viewpoint of dynamics.  \newline
In this work we consider a version of the BCS functional describing a sample of fermionic atoms under the influence of an appropriately localised external potential. We measure the formally infinite free energy of the sample with respect to that of a reference state in the spirit of \cite{HainzlLewinSeiringer,positivedensityliebthirring} and thereby obtain a BCS functional with boundary conditions at infinity. This has to be contrasted with previous works \cite{GL_original,external_field_T_c} where a periodic version of the BCS functional has been studied. Our main result is a lower bound for this BCS functional in terms of expressions that typically appear in Ginzburg-Landau functionals.
\subsection{Set-up}
\label{subsec:the_bcs_functional}
\subsubsection{BCS states and the translation-invariant BCS functional} 
\label{subsect:translation_invariant_BCS}
In BCS theory the quantum mechanical state of a system is most conveniently described by its generalized one-particle density matrix \cite{BachLiebSolovej}. We call an operator $\Gamma \in \mathcal{L}\left( L^2(\mathbb{R}^3) \oplus L^2(\mathbb{R}^3) \right)$ a BCS state if it is of the form
\begin{equation}
\Gamma = \begin{pmatrix} \gamma & \alpha \\ \overline{\alpha} & 1 - \overline{\gamma} \end{pmatrix} \quad \text{ with } \quad 0 \leq \Gamma \leq 1. \label{eq:ch3setupA1}
\end{equation}
Here, $\overline{\alpha} = C \alpha C$ with $C$ denoting complex conjugation. The operators $\gamma$ and $\alpha$ are usually called the one-particle density matrix and the Cooper-pair wave function of the state $\Gamma$, respectively. With the above definitions we cannot conclude that $\gamma$ and $\alpha$ have integral kernels. Nevertheless, admissible states, which we are going to define below, will have this property. In terms of integral kernels, the definition $\overline{\alpha} = C \alpha C$ reads $\overline{\alpha}(x,y) = \overline{\alpha(x,y)}$. For the operators $\gamma$ and $\alpha$, the fact that $\Gamma$ is self-adjoint implies $\gamma^* = \gamma$ and $\alpha^* = \overline{\alpha}$, where the second condition can be rephrased as $\alpha(x,y) = \alpha(y,x)$. The spatial Cooper-pair wave function $\alpha(x,y)$ is symmetric because we do not include spin variables and always assume that Cooper-pairs are in a spin singlet state. This makes the overall Cooper-pair wave function antisymmetric under a combined exchange of position and spin variables. The condition $0 \leq \Gamma \leq 1$ implies  $0 \leq \gamma \leq 1$ as well as $\alpha \overline{\alpha} \leq \gamma \left( 1 - \gamma \right)$.  

In the absence of an external field, it is reasonable to consider translation-invariant states, that is, states with $\alpha(x,y) = \alpha(x-y)$ and $\gamma(x,y) = \gamma(x-y)$. A translation-invariant version of the BCS functional, whose minimizer will play an important role in the definition of our BCS functional with an external field, has been introduced and studied in detail in \cite{BCS_general_pair_interaction}. It reads 
\begin{align}
\mathcal{F}_{\beta}^{ti}(\Gamma) &= \int_{\mathbb{R}^3} (p^2 - \mu) \hat{\gamma}(p) \text{d}p + \int_{\mathbb{R}^3} V(x) \left| \alpha(x) \right|^2 \text{d}x - T S(\Gamma), \label{eq:ch3setupA2} \\
S(\Gamma) &= - \frac{1}{2} \int_{\mathbb{R}^3} \text{Tr}_{\mathbb{C}^2} \left[ \hat{\Gamma}(p) \ln\left(\hat{\Gamma}(p)\right) + (1-\hat{\Gamma}(p)) \ln\left(1-\hat{\Gamma}(p)\right) \right] \text{d}p, \nonumber \\
\hat{\Gamma}(p) &= \begin{pmatrix} \hat{\gamma}(p) & \hat{\alpha}(p) \\ \overline{\hat{\alpha}(p)} & 1 - \overline{\hat{\gamma}(p)} \end{pmatrix}. \nonumber
\end{align}
Here, $T = \beta^{-1} \geq 0$, $\mu \in \mathbb{R}$ and $V \in L^{3/2}(\mathbb{R}^3)$ denote the temperature, the chemical potential and the interaction potential, respectively. It has been shown in \cite{BCS_general_pair_interaction} that there exists a critical temperature $T_{c} \geq 0$ such that for all $T \geq T_c$ the BCS functional is minimized by the pair $\alpha = 0$ and $\gamma(p) = ( 1+\exp[(p^2 - \mu)/(2 T)] )^{-1}$, that is, no Cooper-pairs are present and $\gamma$ is the one-particle density matrix of a free Fermi gas. In contrast, if $T<T_c$ the minimizer will be of the form
\begin{align}
\Gamma_0(p) &= \begin{pmatrix} \gamma_0(p) & \hat{\alpha}_0(p) \\ \overline{\hat{\alpha}_0(p)} & 1 - \overline{\gamma_0(p)} \end{pmatrix} = \frac{1}{1+e^{\beta H_0(p)}}, \label{eq:ch3setupA3} \\
H_0(p) &= \begin{pmatrix} p^2 - \mu & \hat{\Delta}_0(p) \\ \overline{\hat{\Delta}_0(p)} & -(p^2 - \mu) \end{pmatrix}, \nonumber \\
\Delta_0(x) &= 2 V(x) \alpha_0(x) \nonumber
\end{align}
with $\alpha_0 \neq 0$. The three equations Eq.~\eqref{eq:ch3setupA3} are the Euler-Lagrange equations of the translation-invariant BCS functional $\mathcal{F}_{\beta}^{ti}$. The first of the above set of equations can be rewritten in terms of two equations, one for $\alpha_0$ alone and another one that allows us to compute $\gamma_0$ if $\alpha_0$ is given, see \cite{BCS_general_pair_interaction}. The equation for $\alpha_0$ reads
\begin{equation}
\left( K_T^{\Delta_0} + V \right) \alpha_0 = 0, \label{eq:ch3setupA4}
\end{equation}
where
\begin{align}
K_T^{\Delta_0} = \frac{E(-i\nabla)}{\tanh\left( \frac{E(-i\nabla)}{2T} \right)} \quad \text{ and } \quad E(p) = \sqrt{(p^2-\mu)^2 + \left| \hat{\Delta}_0(p) \right|^2}.  \label{eq:ch3setupA5}
\end{align}
Eq.~\eqref{eq:ch3setupA4} can be written as an equation for $\Delta_0$ alone and as such appears in the physics literature under the name gap equation. 
\subsubsection{The BCS functional with boundary conditions at infinity}
\label{ch:The_BCS_functional_and_the_GL_sclaing}
Let $W$ denote an external potential which is assumed to go to zero at infinity. Our goal is to give a meaning to the formal expression
\begin{align}
``\mathcal{F}_\beta(\Gamma) &= \text{Tr}_{L^2(\mathbb{R}^3)} \left[ \left( -\Delta - \mu + W \right) \gamma \right] + \int_{\mathbb{R}^6} V(x-y) \left| \alpha(x,y) \right|^2 \text{d}(x,y) - T S(\Gamma)", \nonumber \\
``S(\Gamma) &= - \frac{1}{2} \text{Tr}_{L^2(\mathbb{R}^3) \oplus L^2(\mathbb{R}^3)} \left[ \Gamma \ln(\Gamma) + (1-\Gamma) \ln(1-\Gamma) \right]", \label{eq:ch3setupB1}
\end{align}
or more precisely to 
\begin{equation}
``\mathcal{F}_{\beta}(\Gamma,\Gamma') = \mathcal{F}_{\beta}(\Gamma) - \mathcal{F}_{\beta}(\Gamma')", \label{eq:ch3setupB2}
\end{equation}
where $\Gamma'$ is a reasonably chosen reference state. The most natural candidate for a reference state is certainly $\Gamma_0$, the minimizer of the translation-invariant BCS functional. (Our assumptions will imply that $\Gamma_0$ is unique up to a phase in front of the Cooper-pair wave function $\alpha_0$.) Nevertheless, it turns out to be useful to work with the state $\Gamma_0^w$ defined by
\begin{align}
\Gamma_{0}^w &= \begin{pmatrix} \gamma_{0}^w & \alpha_{0}^w \\ \overline{\alpha_{0}^w} & 1- \overline{\gamma_{0}^w} \end{pmatrix} = \frac{1}{1+e^{\beta H_{0}^w}},  \label{eq:ch3setupB3} \\
H_{0}^w &= \begin{pmatrix} k(-i \nabla) + W & \hat{\Delta}_0(-i\nabla) \\ \overline{\hat{\Delta}}_0(-i\nabla) & -\left(k(-i \nabla ) + W \right) \end{pmatrix} \nonumber
\end{align} 
and $k(p) = p^2 -\mu$, instead. The advantage of $\Gamma_0^w$ compared to $\Gamma_0$ is that it allows us to include the external potential $W$ in the relative entropy. If we insert $\Gamma_0^w$ into Eq.~\eqref{eq:ch3setupB2} and formally rearrange the outcome, we obtain
\begin{align}
\mathcal{F}_{\beta}\left(\Gamma,\Gamma_0^w\right) &= \frac{1}{2 \beta} \mathcal{H}\left(\Gamma,\Gamma_0^w\right) + \int_{\mathbb{R}6} V(x-y) \left| \alpha(x,y) - \alpha_0^w(x,y) \right|^2 \text{d}(x,y) \label{eq:ch3setupB4} \\
&\hspace{0.5cm} + 2 \text{Re} \int_{\mathbb{R}^6} V(x-y) \left( \alpha(x,y) - \alpha_0^w(x,y) \right) \overline{ \left( \alpha_0^w(x,y) - \alpha_0(x-y) \right) } \text{d}(x,y) \nonumber
\end{align}
where the relative entropy $\mathcal{H}\left(\Gamma,\Gamma_0^w\right)$ of the state $\Gamma$ with respect to the state $\Gamma_0^w$ is defined by
\begin{align}
\mathcal{H}\left(\Gamma,\Gamma_0^w \right) &= \text{Tr} \left[ \varphi(\Gamma) - \varphi(\Gamma_0^w) - \frac{\text{d}}{\text{d}s} \varphi(s\Gamma - (1-s) \Gamma_0^w) \Big|_{s=0} \right], \label{eq:ch3setupB5} \\
\varphi(x) &= x \ln(x) + (1-x) \ln(1-x). \nonumber
\end{align} 
We highlight the term $\frac{\text{d}}{\text{d}s} \varphi(s\Gamma - (1-s) \Gamma_0^w) \Big|_{s=0}$ in Eq.~\eqref{eq:ch3setupB5}. It should be compared with $\varphi'( \Gamma_0^w ) (\Gamma - \Gamma_0^w)$ which usually appears in the definition of the relative entropy. For matrices and trace-class operators both terms yield the same result. For general bounded (and possible non-compact) operators, the definition we use here has more favourable mathematical properties and it appears naturally in the thermodynamic limit, see \cite{Mathieu_Julien,ACR_relative_entropy}. If we had chosen $\Gamma_0$ as a reference state the external potential could not be included in the relative entropy and we would have to deal with a term of the form $\text{Tr}[W(\gamma-\gamma_0)]$ instead. We obt for the other option because we do not know how to control this term.   

We call a BCS state $\Gamma$ \textit{admissible} if $\mathcal{H}(\Gamma,\Gamma_0^w) < \infty$ holds. When we apply the inequality for the relative entropy that has been proven in Lemma~\ref{lem:rel_entropy_inequality} to $\mathcal{H}(\Gamma,\Gamma_0^w)$ and neglect the second term on the right-hand side of Eq.~\eqref{eq:rel_entropy_1}, we obtain
\begin{equation}
\mathcal{H}(\Gamma,\Gamma_0^w) \geq \text{Tr} \left[ (\Gamma-\Gamma_0^w) \frac{\beta H_0^w}{\tanh\left(\frac{\beta H_0^w}{2}\right)} \left( \Gamma- \Gamma_0^w \right) \right] \geq  2 \text{Tr} \left[ (\Gamma-\Gamma_0^w)^2 \right]. \label{eq:ch3setupC2}
\end{equation}
To come to the right-hand side of the above equation, we used $x/\tanh\left(\frac{x}{2T}\right) \geq 2T$. Hence, for any \textit{admissible} state $\Gamma$ the operator $\Gamma-\Gamma_0^w$ is Hilbert-Schmidt. Since $\Gamma_0^w$ has an integral kernel, see Appendix~\ref{Appendix}, this implies that $\Gamma$, and therefore $\gamma$ and $\alpha$, also have integral kernels. In particular, we have $\alpha - \alpha_0^w \in L^2(\mathbb{R}^6)$ and $\gamma - \gamma_0^w \in L^2(\mathbb{R}^6)$. When we assume that $V$ fulfills the assumptions of Lemma~\ref{lem:properties_of_alpha0_W}, it can be applied to control the third term on the right hand side of Eq.~\eqref{eq:ch3setupB4}. Hence, $| \mathcal{F}_{\beta}(\Gamma,\Gamma_0^w)|< \infty$ for any \textit{admissible} state $\Gamma$. By allowing only for \textit{admissible} states we impose the condition that $\alpha(x,y)$ and $\gamma(x,y)$ behave as $\alpha_0(x-y)$ (for one choice of phase) and $\gamma_0(x-y)$ for large values of $|x|$ and/or $|y|$, respectively. In that sense we impose boundary conditions at infinity. 

Our main result is a lower bound for the BCS functional defined in Eq.~\eqref{eq:ch3setupB4} when $T$ is close to (and strictly below) the critical temperature $T_c$ of the translation-invariant BCS functional. To be more precise, we investigate the BCS functional in the scaling that has been introduced in \cite{GL_original} to show the connection between BCS theory and Ginzburg-Landau theory. Hence, we choose $T$ as $T = T_c (1-D h^2)$ for $D>0$ and $h \ll 1$. In macroscopic coordinates, the external potential, the interaction potential and the momentum operator are given by $h^2 W(x)$, $V(x/h)$ and $-ih \nabla$, respectively. The BCS functional in the above scaling reads
\begin{align}
&\mathcal{F}_{\beta}(\Gamma,\Gamma_0^w) = \frac{1}{2 \beta} \mathcal{H}(\Gamma,\Gamma_0^w) + \int_{\mathbb{R}6} V\left(\frac{x-y}{h} \right) \left| \alpha(x,y) - \alpha_0^w(x,y) \right|^2 \text{d}(x,y) \label{eq:ch3setupB6} \\
&\hspace{0.1cm} + 2 \text{Re} \int_{\mathbb{R}^6} V\left( \frac{x-y}{h} \right) \left( \alpha(x,y) - \alpha_0^w(x,y) \right) \overline{ \left( \alpha_0^w(x,y) - h^{-3}\alpha_0\left( \frac{x-y}{h} \right) \right) } \text{d}(x,y), \nonumber
\end{align}
where
\begin{align}
\Gamma_0^w &= \begin{pmatrix} \gamma_0^w & \alpha_0^w \\ \overline{\alpha_0^w} & 1-\overline{\gamma_0^w} \end{pmatrix} = \frac{1}{1+e^{\beta H_0^w}}, \label{eq:ch3setupB7} \\
H_0^w &= \begin{pmatrix} k(-ih\nabla) + h^2 W(x) & \hat{\Delta}_0(-ih \nabla) \\ \hat{\Delta}_0(-ih \nabla) & -\left( k(-ih\nabla) + h^2 W(x) \right) \end{pmatrix}. \nonumber
\end{align}
This is the version of the BCS functional we are going to study in the rest of this text.
\subsection{Main results}
Before we state our results, let us make a few assumptions. In order to be able to carry out computations in a convenient way, we assume some regularity conditions for our potentials $V$ and $W$. By $H^n(\mathbb{R}^3)$ and $W^{n,p}(\mathbb{R}^3)$ we denote the usual Sobolev spaces equipped with their natural norms.
\begin{assumption}
\label{assumption1}
\textit{We assume for the interaction potential $V$ that $V \in H^1(\mathbb{R}^3) \cap W^{1,\infty}(\mathbb{R}^3)$ together with $\hat{V} \in L^1(\mathbb{R}^3) \cap H^4(\mathbb{R}^3) \cap W^{2,\infty}(\mathbb{R}^3)$. The external potential $W$ obeys $W \in H^1(\mathbb{R}^3) \cap W^{1,\infty}(\mathbb{R}^3)$ with $\hat{W} \in L^1(\mathbb{R}^3) \cap W^{4,1}(\mathbb{R}^3) \cap W^{4,\infty}(\mathbb{R}^3)$ and $(1+(\cdot)^2) \hat{W} \in L^{\infty}(\mathbb{R}^3)$. Additionally, we assume that $V$ and $W$ are symmetric functions, that is $V(-x) = V(x)$ and $W(-x)=W(x)$ for all $x \in \mathbb{R}^3$. }
\end{assumption}
Most of the above assumptions could be relaxed, but we rather prefer to keep the proofs to a reasonable length. On the other hand, the following assumptions for the interaction potential $V$ are crucial.   
\begin{assumption}
\textit{The potential $V$ is such that the following two statements are true: (i) $T_c > 0$, (ii) The ground state of the operator $K^{0}_{T_c} + V$ is non-degenerate.}
\end{assumption}
\begin{remark}
\label{remark1}
Property (i) holds for example if $\mu > 0$ and $V \in L^{3/2}(\mathbb{R}^3)$ is negative and not identically zero, see \cite[Theorem~3]{BCS_general_pair_interaction}. The same Theorem also tells us that one can add a positive part $V_+$ to the potential $V$ and thereby keep the property $T_c > 0$ if $V_+$ is small enough in a suitable sense.
\end{remark}
\begin{remark}
\label{remark3}
Property (ii) implies the uniqueness of the minimizer of the translation-invariant BCS functional, see \cite[Section~4.6]{reviewBCS}. It holds for example if $\hat{V} \leq 0$ with $\hat{V} \not\equiv 0$. The assumption also implies that there exists a number $0\leq T'<T_c$ such that for all $T \in [T',T_c)$, $\alpha_0(T)$ is the unique ground state of the operator $K_T^{\Delta_0} + V$ \cite[Section~4.6]{reviewBCS}.
\end{remark}
To state our result, we have to introduce a decomposition of the Cooper-pair wave function $\alpha$. Let the measurable function $\psi$ be given by 
\begin{equation}
\psi(y) = \frac{\int_{\mathbb{R}^3} \alpha_0\left( \frac{x-y}{h} \right) \alpha(x,y) \text{d}x}{\int_{\mathbb{R}^3} \left| \alpha_0(x) \right|^2 \text{d}x}. \label{eq:ch3setupD3} 
\end{equation}
We note that $\psi$ depends on temperature because $\alpha_0$ does. Then $\alpha$ can be written as
\begin{equation}
\alpha(x,y) = h^{-3} \alpha_0\left( \frac{x-y}{h} \right) \frac{\psi(x) + \psi(y)}{2} + \xi(x,y).
\label{eq:ch3setupD4}
\end{equation}
A similar decomposition also plays an important role in the derivation of Ginzburg-Landau theory in \cite{GL_original}. 

Our main theorem is a lower bound for the BCS functional defined in Eq.~\eqref{eq:ch3setupB6} which implies a-priori estimates for states $\Gamma$ with energy less than or equal to that of the reference state.
\begin{theorem}
\label{thm:lower_bound_BCS}
Let $\mathcal{F}_{\beta}$ be given by Eq.~\eqref{eq:ch3setupB6} with $\beta^{-1}=T=T_c (1-Dh^2)$, $D>0$ and let $\Gamma$ be an \textit{admissible} state with $\mathcal{F}_{\beta}\left( \Gamma, \Gamma_0^w \right) \leq 0$. Then for $h>0$ small enough, there exist constants $C_1, C_2 >0$ such that
\begin{align}
\mathcal{F}_{\beta}(\Gamma,\Gamma_0^w) &\geq C_1 \Bigg( h \left\Vert \nabla \psi \right\Vert_{L^2(\mathbb{R}^3)}^2 + h \left\Vert |\psi|^2-1 \right\Vert_{L^2(\mathbb{R}^3)}^2 + \left\Vert \xi \right\Vert_{H_1(\mathbb{R}^6)}^2 \label{eq:ch3setupD5} \\
&\hspace{7cm} + \left\Vert \gamma - \gamma_0^w \right\Vert_{H^1(\mathbb{R}^6)}^2 \Bigg) - C_2 h. \nonumber
\end{align}
In the above equation, the $H^1(\mathbb{R}^6)$-norms are, according to our choice of coordinates, given by $\left\Vert f \right\Vert_{H^1(\mathbb{R}^6)}^2 = \left\Vert f \right\Vert_{L^2(\mathbb{R}^6)}^2 + \left\Vert h \nabla_x f \right\Vert_{L^2(\mathbb{R}^6)}^2 + \left\Vert h \nabla_y f \right\Vert_{L^2(\mathbb{R}^6)}^2$. Eq.~\eqref{eq:ch3setupD5} implies the a-priori bounds
\begin{align}
\left\Vert \nabla \psi \right\Vert_{L^2(\mathbb{R}^3)} + \left\Vert |\psi|^2-1 \right\Vert_{L^2(\mathbb{R}^3)} &\leq C, \label{eq:ch3setupD6} \\
\left\Vert \xi \right\Vert_{H^1(\mathbb{R}^6)} + \left\Vert \gamma - \gamma_0^w \right\Vert_{H^1(\mathbb{R}^6)} &\leq C h^{1/2}, \nonumber 
\end{align}
for an appropriately chosen constant $C>0$.
\end{theorem}
\begin{remark}
\label{remark4}
The construction of a lower bound for the periodic version of the BCS functional investigated in \cite{GL_original,external_field_T_c} is not difficult and goes along the same lines as the one for the translation-invariant BCS functional. A crucial ingredient in this case is the relation $\alpha \overline{\alpha} \leq \gamma (1-\gamma)$ which cannot be used in our problem. We treat the formal difference of two BCS functionals, and hence one inequality would always go in the wrong direction. This makes the same problem for our functional much harder. Although infinite volume functionals have been considered in the situation where states can be described by infinite volume versions of slater determinants, see e.g. \cite{HainzlLewinSeiringer,positivedensityliebthirring}, this is to our knowledge the first time that such a functional has been considered for quasi-free states with a pairing function. 
\end{remark}
\begin{remark}
\label{remark6}
The first two terms on the right hand side of Eq.~\eqref{eq:ch3setupD5} already look like two expressions that typically appear in Ginzburg-Landau functionals which (without a magnetic field) take the form
\begin{equation}
\mathcal{E}(\psi) = \int_{\mathbb{R}^3} \left[ \overline{\nabla \psi(x)} B_1 \nabla \psi(x) + B_2 W(x) \left| \psi(x) \right|^2 + B_3 \left| 1- \left| \psi(x) \right|^2 \right|^2 \right] \text{d}x.
 \label{eq:ch3setupD7}
\end{equation}
Here $B_1$ is some positive $3 \times 3$ matrix, $B_2 \in \mathbb{R}$ and $B_3 > 0$. \newline
The a-priori estimates for states with energy less than or equal to that of the reference state in Theorem~\ref{thm:lower_bound_BCS} should be compared with \cite[Chapter~5]{GL_original}, where related estimates for states with energy less than or equal to that of the normal state are derived. As the estimates in \cite[Chapter~5]{GL_original} are the first step in the derivation of Ginzburg-Landau theory, the a-priori estimates proved here can be the starting point for a similar analysis in our set-up.  Formal computations indicate that a similar Theorem as \cite[Theorem~1]{GL_original} also holds in our case if $B_1,B_2$ and $B_3$ are chosen like in \cite[Eqs.~(1.19)-(1.21)]{GL_original}. \newline
The remaining steps in the hypothetical proof of this conjecture are far from being straightforward, however. Let us mention what we consider to be main sources of difficulties: An important ingredient used in the proof of the upper and the lower bound in \cite{GL_original} is the a-priori bound for the $L^2(\mathbb{R}^3)$-norm of $\psi$. On the one hand, it serves well in error estimates and on the other hand it allows one to use techniques based on Fourier analysis. In our situation, $\psi-1$, which appears naturally in the formulas as a replacement for $\psi$, does not have a bounded $L^2(\mathbb{R}^3)$-norm, we control only $\left\Vert |\psi|^2-1 \right\Vert_{L^2(\mathbb{R}^3)}$. Therefore, the method introduced in \cite{GL_original} does not work for our set-up without major adjustments. Additionally, the fact that typical BCS states are non-compact in our setting complicates the analysis at several places considerably. Compare e.g. the proof of \cite[Lemma~1]{GL_original} with the proof of the same statement in our work, see Lemma~\ref{lem:rel_entropy_inequality} and the definition of the relative entropy, where non-trivial new ideas have been needed. It is to be expected that similar technical problems also arise in the remaining steps of the proof of the above conjecture. \newline
Our Theorem~\ref{thm:lower_bound_BCS} therefore is the first necessary and important step towards an extension of the microscopic derivation of Ginzburg-Landau theory to the physically more realistic situation of a whole-space functional. 
\end{remark}  
\section{Proof of Theorem~\ref{thm:lower_bound_BCS}}
\label{sec:lower_bound_functional}
Our construction of a lower bound for the BCS functional will be done in three steps. In the first step, we prove an inequality for the relative entropy that has been introduced in \cite{GL_original} in the setting of periodic states and will enable us to control the quadratic interaction term in the BCS functional, that is, the second term on the right-hand side of Eq.~\eqref{eq:ch3setupB6}. This inequality is harder to prove in our setting because our states are in general non-compact operators. The proof uses a recent approximation result for a family of generalized relative entropies for bounded operators \cite{Mathieu_Julien,ACR_relative_entropy}. What remains to be controlled after this step are the terms proportional to the external potential $h^2 W$ and the linear interaction term, which is the third term on the right-hand side of Eq.~\eqref{eq:ch3setupB6}. In the second step, we derive bounds on these non-positive terms. Whenever we encounter terms of the form $\text{Tr}\left[ (\gamma-\gamma_0^w) h^2 W (\gamma-\gamma_0^w) \right]$, it is sufficient to use the fact that $h^2 \left\Vert W \right\Vert_{L^{\infty}(\mathbb{R}^3)}$ is small. In contrast, if $h^2 W$ is acting on $\alpha - \alpha_0^w$ we explicitly have to use that $W$ is a function that goes to zero at infinity. In the third and final step, we use the bounds derived in step~1 and step~2 to show that the BCS functional is bounded from below. \newline
During our proof we will frequently encounter the situation where we need to know that certain norms of $\alpha_0$ are of order $h$. This is guaranteed by Lemma~\ref{lem:properties_of_alpha_0_1} and \cite[Proposition~5.6]{Frank_Lemm}. In \cite[Proposition~5.6]{Frank_Lemm} the authors show that for $T<T_c$ with $T$ close to $T_c$ any minimizer of the translation-invariant BCS functional is of the form $\alpha_0=\alpha_{*} + \eta$, where $\alpha_{*}$ lies in the kernel of $K_{T_c}^0+V$ and obeys $\left\Vert \alpha_{*} \right\Vert_{L^2(\mathbb{R}^3)} = O(\sqrt{T_c-T})$. The function $\eta$ is an element of the orthogonal complement of the kernel of $K_{T_c}^0+V$ and behaves as $\left\Vert \eta \right\Vert_{L^2(\mathbb{R}^3)} = O(T_c-T)$.  Together with Assumptions~\ref{assumption1} and Lemma~\ref{lem:properties_of_alpha_0_1}, this implies that all norms of $\alpha_0$ we will encounter are of order $h$. In order to not distract from the main ideas, we will sometimes use these facts without further reference. 
\subsection{Step 1: Lower bound for the relative entropy and domination of the quadratic interaction term} 
\label{subsec:lower_bound_step1}
We start our discussion with a version of an inequality for the relative entropy that has been introduced in \cite{GL_original} in the setting of periodic states. A similar bound without the equivalent of the second positive term on the right hand side of Eq.~\eqref{eq:rel_entropy_1} was used in \cite{HainzlLewinSeiringer}.
\begin{lemma}
\label{lem:rel_entropy_inequality}
Let $\Gamma$ and $\Gamma'$ be BCS states such that $\Gamma' = \left( 1+e^{H} \right)^{-1}$ for some self-adjoint operator $H$ on $L^2(\mathbb{R}^3) \oplus L^2(\mathbb{R}^3)$. Then the inequality
\begin{equation}
\mathcal{H}(\Gamma,\Gamma') \geq \text{Tr} \left[ \left( \Gamma - \Gamma' \right) \frac{H}{\tanh \left( H/2 \right) } \left( \Gamma - \Gamma' \right) \right] + \frac{4}{3} \ \text{Tr} \left[ \Gamma \left( 1-\Gamma \right) - \Gamma' \left( 1 - \Gamma' \right) \right]^2
\label{eq:rel_entropy_1}
\end{equation}
holds.
\end{lemma}
\begin{proof}
In \cite[Lemma 1]{GL_original} it has been shown that for any pair of real numbers $0<x$, $y<1$, one has
\begin{equation}
x \ln \left( \frac{x}{y} \right) + (1-x) \ln \left( \frac{1-x}{1-y} \right) \geq \frac{\ln \left( \frac{1-y}{y} \right)}{1-2y}(x-y)^2 + \frac{4}{3} \left( x(1-x) - y(1-y) \right)^2.
\label{eq:rel_entropy_2}
\end{equation}
Let $\lbrace P_n \rbrace_{n=1}^{\infty}$ be an increasing sequence of orthogonal finite dimensional projections with $P_n \rightarrow 1$ in the strong operator topology of $L^2(\mathbb{R}^3) \oplus L^2(\mathbb{R}^3)$. Using Klein's Lemma, see e.g. \cite[Section~2.1.4]{Thirring_4}, Eq.~\eqref{eq:rel_entropy_2} implies
\begin{align}
\mathcal{H}(\Gamma_n,\Gamma_{n}') &\geq \text{Tr} \left[ \left( \Gamma_n - \Gamma_{n}' \right) \frac{\ln \left( \frac{1-\Gamma_{n}'}{\Gamma_{n}'} \right)}{1-2 \Gamma_{n}'} \left( \Gamma_n -  \Gamma_{n}' \right) \right] \label{eq:rel_entropy_3} \\
&\hspace{4cm} + \frac{4}{3} \ \text{Tr} \left[ \Gamma_n \left( 1- \Gamma_{n} \right) - \Gamma_{n}' \left( 1 - \Gamma_{n}' \right) \right]^2, \nonumber
\end{align}
where $\Gamma_n = P_n \Gamma P_n$ and $\Gamma_n'=P_n \Gamma' P_n$. By \cite[Theorem~3]{ACR_relative_entropy} the left-hand side of Eq.~\eqref{eq:rel_entropy_3} converges to $\mathcal{H}(\Gamma,\Gamma')$ as $n$ tends to infinity.

Since $P_n \xrightarrow{\text{s}} 1$ strongly, we know that $\left[ \Gamma_n \left( 1- \Gamma_{n} \right) - \Gamma_{n}' \left( 1 - \Gamma_{n}' \right) \right]^2 \xrightarrow{\text{s}} \left[ \Gamma \left( 1- \Gamma \right) - \Gamma' \left( 1 - \Gamma' \right) \right]^2$. Fatou's Lemma for traces \cite[Theorem~2.7]{Simon_trace_ideals} yields
\begin{align}
\liminf_{n \rightarrow \infty} \text{Tr} \left[ \Gamma_n \left( 1- \Gamma_{n} \right) - \Gamma_{n}' \left( 1 - \Gamma_{n}' \right) \right]^2 \geq \text{Tr} \left[ \Gamma \left( 1- \Gamma \right) - \Gamma' \left( 1 - \Gamma' \right) \right]^2.
\label{eq:rel_entropy_4}
\end{align}
The same strategy cannot be applied for the first term on the right-hand side of Eq.~\eqref{eq:rel_entropy_3} because the operator under the trace in the middle may become unbounded. To obtain a similar result for this term, we choose an orthonormal basis  $\lbrace e_{\alpha} \rbrace_{\alpha =1 }^{\infty}$ of $L^2(\mathbb{R}^3) \oplus L^2(\mathbb{R}^3)$ and use the classical Fatou's Lemma.
\begin{align}
\liminf_{n \rightarrow \infty} \sum_{\alpha = 1}^{\infty} &\left( e_{\alpha}, \left( \Gamma_n - \Gamma_{n}' \right) \frac{\ln \left( \frac{1-\Gamma_{n}'}{\Gamma_{n}'} \right)}{1-2 \Gamma_{n}'} \left( \Gamma_n -  \Gamma_{n}' \right) e_{\alpha} \right) \label{eq:rel_entropy_5} \\
&\hspace{3cm} \geq \sum_{\alpha = 1}^{\infty} \liminf_{n \rightarrow \infty} \left( e_{\alpha}, \left( \Gamma_n - \Gamma_{n}' \right) \frac{\ln \left( \frac{1-\Gamma_{n}'}{\Gamma_{n}'} \right)}{1-2 \Gamma_{n}'} \left( \Gamma_n -  \Gamma_{n}' \right) e_{\alpha} \right). \nonumber
\end{align}
Let $\mu_{n \alpha}$ be the spectral measure of the operator $\Gamma_{n}'$ with respect to the vector $\left( \Gamma_{n}- \Gamma_{n}' \right) e_{\alpha}$. The expectation values in Eq.~\eqref{eq:rel_entropy_5} can be written as
\begin{equation}
\left( \left( \Gamma_n - \Gamma_{n}' \right) e_{\alpha}, \frac{\ln \left( \frac{1-\Gamma_{n}'}{\Gamma_{n}'} \right)}{1-2 \Gamma_{n}'} \left( \Gamma_n -  \Gamma_{n}' \right) e_{\alpha} \right) = \int_0^1 \frac{\ln \left( \frac{1-\lambda}{\lambda} \right)}{1-2 \lambda} \text{d}\mu_{n \alpha}(\lambda).
\label{eq:rel_entropy_6}
\end{equation}
Correspondingly, we denote by $\mu_{\alpha}$ the spectral measure of the operator $\Gamma'$ with respect to the vector $\left( \Gamma- \Gamma' \right) e_{\alpha}$. It can easily be seen that $\mu_{n \alpha} \rightharpoonup \mu_{\alpha}$ for all $\alpha \in \mathbb{N}$ in the weak-$\ast$ topology of the space of Borel measures. To obtain a lower bound for the integral on the right-hand side of Eq.~\eqref{eq:rel_entropy_6}, we use Fatou's Lemma for sequences of measures, see e.g.~\cite[Theorem~30.2]{Bauer}, which implies
\begin{equation}
\liminf_{n \rightarrow \infty} \int_0^1 \frac{\ln \left( \frac{1-\lambda}{\lambda} \right)}{1-2 \lambda} \text{d}\mu_{n \alpha}(\lambda) \geq \int_0^1 \frac{\ln \left( \frac{1-\lambda}{\lambda} \right)}{1-2 \lambda} \text{d}\mu_{\alpha}(\lambda). \label{eq:rel_entropy_7}
\end{equation}
A few algebraic manipulations show $\frac{\ln \left( \frac{1-\Gamma'}{\Gamma'} \right)}{1-2 \Gamma'} = \frac{H}{\tanh(H/2)}$ and conclude our proof.
\end{proof}
When we apply Lemma~\ref{lem:rel_entropy_inequality} to $\mathcal{H}(\Gamma,\Gamma_0^w)$ we obtain two terms. Following the strategy of \cite[Chapter~5]{GL_original}, we now derive a lower bound for the equivalent of the last term on the right hand side of Eq.~\eqref{eq:rel_entropy_1}. When we write the traces of the operator-valued matrices explicitly in terms of their components, it can easily be seen that 
\begin{equation}
\text{Tr} \left[ \Gamma \left( 1-\Gamma \right) - \Gamma_{0}^w \left( 1 - \Gamma_{0}^w \right) \right]^2 \geq \ 2 \text{Tr} \left[ \gamma (1-\gamma) - \gamma_0^w (1-\gamma_0^w) - \alpha \overline{\alpha} + \alpha_0^w \overline{\alpha_0^w} \right]^2
\label{eq:step1_9}
\end{equation}
holds. We also claim that
\begin{align}
&2 \text{Tr} \left(\gamma - \gamma_0^w \right)^2 + \frac{4}{3} \ \text{tr}\left[ \gamma (1-\gamma) - \gamma_0^w \left( 1- \gamma_0^w \right) - \alpha \overline{\alpha} + \alpha_0^w \overline{\alpha_0^w} \right]^2 \label{eq:step1_10} \\
&\hspace{8cm}\geq \frac{4}{5} \ \text{Tr}\left( \alpha \overline{\alpha} - \alpha_0^w \overline{\alpha_0^w} \right)^2, \nonumber
\end{align}
which follows from
\begin{align}
\left\Vert \alpha \overline{\alpha} - \alpha_0^w \overline{\alpha}_0^w \right\Vert_2 &\leq \left\Vert \gamma (1-\gamma) - \gamma_0^w (1-\gamma_0^w) - \left( \alpha \overline{\alpha} - \alpha_0^w \overline{\alpha}_0^w \right) \right\Vert_2 \label{eq:step1_11} \\
&\hspace{6cm} + \left\Vert \gamma (1-\gamma) - \gamma_0^w \left( 1-\gamma_0^w \right) \right\Vert_2 \nonumber
\end{align}
together with 
\begin{equation}
\left\Vert \gamma (1-\gamma) - \gamma_0^w \left( 1- \gamma_0^w \right) \right\Vert_2 \leq \left\Vert \gamma - \gamma_0^w \right\Vert_2.
\label{eq:step1_12}
\end{equation}
In the above formulas, $\left\Vert \cdot \right\Vert_2$ denotes the Hilbert-Schmidt norm.  Eq.~\eqref{eq:step1_12} can most conveniently be seen to hold as follows. Choose an increasing sequence of orthogonal finite rank projections $\left\lbrace P_n \right\rbrace_{n=1}^{\infty}$ acting on $L^2(\mathbb{R}^3)$ that converges to the identity in the strong operator topology. Using Klein's inequality, one easily shows that
\begin{equation}
\text{Tr} \left( \gamma_n \left( 1 - \gamma_n \right) - \gamma_{0,n}^w \left( 1 - \gamma_{0,n}^w \right) \right)^2 \leq \text{Tr} \left( \gamma_n - \gamma_{0,n}^w \right)^2,
\label{eq:step1_13}
\end{equation} 
where we have introduced the notation $\gamma_n = P_n \gamma P_n$ and $\gamma_{0,n}^w = P_n \gamma_0^w P_n$. With arguments similar to the ones used in the proof of Lemma~\ref{lem:rel_entropy_inequality}, one now easily shows the claim.

Now we come to the second term that we obtain after Lemma~\ref{lem:rel_entropy_inequality} has been applied to $\mathcal{H}(\Gamma,\Gamma_0^w)$. It reads
\begin{equation}
\text{Tr} \left[ \left( \Gamma - \Gamma_0^w \right) \frac{H_0^w}{\tanh\left( \frac{\beta H_0^w}{2} \right)} \left( \Gamma - \Gamma_0^w \right) \right]. \label{eq:step1_1}
\end{equation}
In order to write this expression as a sum of $\text{Tr}[(\Gamma-\Gamma_0^w) K_T^{\Delta_0} (\Gamma-\Gamma_0^w)]$, which we will need to control the second term on the right-hand side of Eq.~\eqref{eq:ch3setupB6}, plus corrections that are proportional to $h^2 W$, we use the respresentation \cite[(4.3.91)]{Abramowitz_Stegun},\cite{GL_original},
\begin{equation}
\frac{H_0^w}{\tanh \left( \frac{H_0^w}{2 T} \right) } = 2T + 2T \sum_{n=1}^{\infty} \left( 1 - \frac{c^2 n^2}{\left( H_0^w \right)^2 + c^2 n^2 } \right), \label{eq:step1_3}
\end{equation}
with $c = 2 \pi T$.  Applying resolvent identities, we obtain a decomposition of $H_0^w/\tanh(\beta H_0^w/2)$ of the form
\begin{equation}
\frac{H_0^w}{\tanh\left( \frac{\beta H_0^w}{2} \right)} = \mathds{1}_{\mathbb{C}^2} K_{T}^{\Delta_0} + A + B,
\label{eq:step1_5}
\end{equation}
where the operators $A$ and $B$ are given by
\begin{align}
A &= \sum_{n=1}^{\infty} \frac{1}{E^2 + c^2 n^2} \left( H_0 h^2 \omega + h^2 \omega H_0 + h^4 \omega^2  \right) \frac{1}{E^2 + c^2 n^2}, \label{eq:step1_6} \\
B &= -\sum_{n=1}^{\infty} \frac{1}{E^2 + c^2 n^2} \left( H_0 h^2 \omega + h^2 \omega H_0 + h^4 \omega^2  \right) \frac{1}{\left( H_0^w \right)^2 + c^2 n^2} \nonumber \\
&\ \ \hspace{4cm} \times \left( H_0 h^2 \omega + h^2 \omega H_0 + h^4 \omega^2  \right) \frac{1}{E^2 + c^2 n^2}, \nonumber
\end{align}
and
\begin{equation}
\omega = \begin{pmatrix} W & 0 \\ 0 & -W \end{pmatrix}.
\end{equation}
To obtain the above result, we used on the one hand that $x \mapsto x/\tanh(x/2T)$ is an even function of $x$ and on the other hand that $H_0^2 = \mathds{1}_{\mathbb{C}^2} E^2$ holds.

Since we will have to deal with $\Gamma - \Gamma_0^w$ frequently, we introduce the following notation:
\begin{equation}
Q = \Gamma - \Gamma_0^w, \quad \quad \Lambda = \alpha - \alpha_0^w \quad \text{ and } \quad q = \gamma - \gamma_0^w. \label{eq:step1_7}
\end{equation}
When we explicitly evaluate the trace over the $\mathbb{C}^2$-matrix structure in the term $\text{Tr}\left[ Q K_T^{\Delta_0} Q \right]$ and use $K_T^{\Delta_0} \geq 2T$ as well as $K_T^{\Delta_0} \geq C(T)(1 + p^2)$, we arrive at the following lower bound for the BCS functional:
\begin{align}
\mathcal{F}_{\beta}(\Gamma,\Gamma_0^w) &\geq \int_{\mathbb{R}^3} \left( \Lambda, \left( K_{T,x}^{\Delta_0} + V_y \right) \Lambda \right)_{L^2(\mathbb{R}^3,\text{d}x)} \text{d}y + C \left\Vert q \right\Vert_{H^1(\mathbb{R}^6)}^2 \label{eq:step1_17} \\
&\hspace{2cm} + \frac{1}{2} \text{Tr} \left[ Q (A+B) Q \right] + \frac{2}{5 \beta} \text{Tr} \left( \alpha \overline{\alpha} - \alpha_0^w \overline{\alpha_0^w} \right)^2 \nonumber \\
&\hspace{2cm} + 2 \text{Re} \int_{\mathbb{R}^6} V \left( \frac{x-y}{h} \right) \Lambda(x,y)  \overline{\tilde{\alpha}_0^w}(x,y) \text{d}(x,y). \nonumber
\end{align}
In Eq.~\eqref{eq:step1_17} we write $V_y(x) = V(x-y)$ and $\tilde{\alpha}_0^w(x,y) = \alpha_0^w(x,y) - h^{-3} \alpha_0\left( \frac{x-y}{h} \right)$. The subscript $x$ in the operator $K_{T,x}^{\Delta_0}$ tells us that it is acting on the $x$-component of the function $\Lambda(x,y)$. The remaining part of this section is devoted to the question which norms related to $\alpha$ can be controlled by the first term on the right hand side of Eq.~\eqref{eq:step1_17}.

The ground state of the operator $K_{T_c,x}^{0}+V_y$ is unique by assumption and $V \in L^2(\mathbb{R}^3)$ guarantees the existence of a spectral gap between the ground state and the first excited state which we denote by $\kappa > 0$. Since $\Vert \hat{\Delta}_0 \Vert_{L^{\infty}(\mathbb{R}^3)} = O(h)$, perturbation theory tells us that the second eigenvalue of the operator $K_T^{\Delta_0} + V_y$ can be bounded from below by $\kappa/2$. Our assumptions also tell us that $K_{T,x}^{\Delta_0} + V_y$ is a positive operator and that the only elements in its kernel are functions of the form $\alpha_0\left( \frac{x-y}{h} \right) \psi(y)$. Let us introduce the abbreviation $\alpha_0 \psi$ to denote the operator $\hat{\alpha}_0( -ih \nabla) \psi(x)$.  We insert the decomposition $\alpha = \alpha_0 \psi + \xi_0$ into the first term on the right-hand side of Eq.~\eqref{eq:step1_17} and obtain
\begin{align}
& \int_{\mathbb{R}^3} \left( \alpha_0 \psi + \xi_0 - \alpha_0^w, \left( K_{T,x}^{\Delta_0} + V_y \right) \left( \alpha_0 \psi + \xi_0 - \alpha_0^w \right) \right)_{L^2(\text{d}x)} \text{d}y \label{eq:step1_21} \\
&\hspace{6cm} \geq C_1 \left\Vert \xi_0 \right\Vert_{H_x^1(\mathbb{R}^6)}^2 - C_2 \left\Vert \xi_0 \right\Vert_{H_x^1(\mathbb{R}^6)} \left\Vert \tilde{\alpha}_0^w \right\Vert_{H_x^1(\mathbb{R}^6)} \nonumber \\
&\hspace{6cm} \geq C_1' \left\Vert \xi_0 \right\Vert_{H_x^1(\mathbb{R}^6)}^2 - C_2' h^{3/2}. \nonumber
\end{align}
By $\left\Vert \cdot \right\Vert_{H_x^1(\mathbb{R}^6)}$ we denote that part of the $H^1(\mathbb{R}^6)$-norm where the derivatives with respect to the $y$-coordinate are dropped. To obtain the result we used $c  (1+(hp)^2) \leq K_T^{\Delta_0} \leq C (1+(hp)^2)$ as well as Lemma~\ref{lem:properties_of_alpha0_W}. 

On the other hand, the left hand side of Eq.~\eqref{eq:step1_21} can be bounded from below in terms of \newline $h^2 \left\Vert (\nabla_x + \nabla_y)\Lambda \right\Vert_{L^2(\mathbb{R}^6)}^2$ as the following Lemma shows:
\begin{lemma}
\label{lem:lower_bound_K_T_V}
Let $\Lambda \in H^1(\mathbb{R}^6)$ be a symmetric function [$\Lambda(x,y)=\Lambda(y,x)$]. Then there exists a constant $C>0$ such that
\begin{equation}
\int_{\mathbb{R}^3} \left( \Lambda, \left( K_T^{\Delta_0} + V_y \right) \Lambda\right)_{L^2(\text{d}x)} \text{d}y \geq C h^2 \int_{\mathbb{R}^6} \left| \left( \nabla_x + \nabla_y \right) \Lambda(x,y) \right|^2 \text{d}(x,y). \label{eq:step1_22}
\end{equation}
\end{lemma}
\begin{proof}
The proof goes along the same lines as the one given in \cite[Lemma 3]{GL_original}.
\end{proof}

To make a connection between the decomposition $\alpha = \alpha_0 \psi + \xi_0$ and Lemma~\ref{lem:lower_bound_K_T_V}, we consult Eq.~\eqref{eq:ch3setupD3}. Using integration by parts and $\alpha_0(-x) = \alpha_0(x)$, which is assured by our assumptions, we find
\begin{equation}
\nabla \psi(y) = \frac{\int_{\mathbb{R}^3} \alpha_0 \left( \frac{x-y}{h} \right) \left[ (\nabla_x + \nabla_y) \alpha(x,y) \right] \text{d}x}{\left\Vert \alpha_0 \right\Vert_{L^2(\mathbb{R}^3)}^2}. \label{eq:step1_23} 
\end{equation}
Next, we integrate Eq.~\eqref{eq:step1_23} over $y$ and apply Schwarz's inequality once which yields
\begin{equation}
\left\Vert \nabla \psi \right\Vert_{L^2(\mathbb{R}^3)}^2 \leq h^3 \frac{\left\Vert \left( \nabla_x + \nabla_y \right) \alpha(x,y) \right\Vert_{L^2(\mathbb{R}^6)}^2}{\left\Vert \alpha_0 \right\Vert_{L^2(\mathbb{R}^3)}^2}. \label{eq:step1_24}
\end{equation}
Together with Lemma~\ref{lem:lower_bound_K_T_V} and Lemma~\ref{lem:properties_of_alpha0_W}, Eq.~\eqref{eq:step1_24} implies
\begin{equation}
\int_{\mathbb{R}^3} \left( \Lambda, \left( K_T^{\Delta_0} + V_y \right) \Lambda\right)_{L^2(\text{d}x)} \text{d}y \geq h^{-1} C_1 \left\Vert \alpha_0 \right\Vert_{L^2(\mathbb{R}^3)}^2 \left\Vert \nabla \psi \right\Vert^2 - C_2 h^{5}. \label{eq:step1_26}
\end{equation}
To also control the $L^2(\mathbb{R}^6)$-norm of $h \nabla_y \xi_0(x,y)$, we use the triangle inequality and estimate
\begin{align}
&\left\Vert (\nabla_x + \nabla_y) \alpha(x,y) \right\Vert_{L^2(\mathbb{R}^6)}^2 \geq \frac{1}{2} \left\Vert (\nabla_x + \nabla_y) \xi_0(x,y) \right\Vert_{L^2(\mathbb{R}^6)}^2 \label{eq:step1_27} \\
&\hspace{8cm} - \left\Vert \alpha_0 \right\Vert_{L^2(\mathbb{R}^3)}^2 \left\Vert \nabla \psi \right\Vert_{L^2(\mathbb{R}^3)}^2. \nonumber
\end{align} 
When we apply \cite[Proposition~5.6]{Frank_Lemm} which tells us that $\left\Vert \alpha_0 \right\Vert_{L^2(\mathbb{R}^3)} \lesssim h$ and put the results of this paragraph together, we finally arrive at
\begin{equation}
\int_{\mathbb{R}^3} \left( \Lambda, \left( K_T^{\Delta_0} + V_y \right) \Lambda\right)_{L^2(\text{d}x)} \text{d}y \geq C_1 \left( h \left\Vert \nabla \psi \right\Vert_{L^2(\mathbb{R}^3)}^2 + \left\Vert \xi_0 \right\Vert_{H^1(\mathbb{R}^6)}^2 \right) - C_2 h^{3/2}. \label{eq:step1_28} 
\end{equation}
Using the symmetry of $\Lambda$, that is $\Lambda(x,y) = \Lambda(y,x)$, and $V(x-y)=V(y-x)$, we can write the first term on the right hand side of Eq.~\eqref{eq:step1_17} with $K_{T,x}^{\Delta_0} + V_y$ replaced by $K_{T,y}^{\Delta_0} + V_x$ (now both operators act on the $y$-component of $\Lambda(x,y)$). Let us write $\alpha=\psi \alpha_0 + \xi_1$ where $\psi \alpha_0$ denotes the operator $\psi(x) \hat{\alpha}_0( -ih \nabla)$ and go through the whole analysis until Eq.~\eqref{eq:step1_28} again. We will find a similar bound as Eq.~\eqref{eq:step1_28} with $\xi_0$ replaced by $\xi_1$. Both decompositions of $\alpha$ we have used so far are not symmetric in $x$ and $y$. In a last step we write $\alpha=\left( \alpha_0 \psi + \psi \alpha_0 \right)/2 + \xi$ and realize that $\xi$ is given by $\xi = \left( \xi_0 + \xi_1 \right)/2$. This implies $\left\Vert \xi \right\Vert_{H^1(\mathbb{R}^6)} \leq ( \left\Vert \xi_0 \right\Vert_{H^1(\mathbb{R}^6)} + \left\Vert \xi_1 \right\Vert_{H^1(\mathbb{R}^6)} )/2$, which means that we can also obtain the $H^1(\mathbb{R}^6)$-norm of $\xi$ in our lower bound.

Insertion of our results into Eq.~\eqref{eq:step1_17} gives 
\begin{align}
&\mathcal{F}_{\beta}(\Gamma,\Gamma_0^w)  \geq C_1 \left( h \left\Vert \nabla \psi \right\Vert_{L^2(\mathbb{R}^3)}^2 + \left\Vert \xi \right\Vert_{H^1(\mathbb{R}^6)}^2 + \left\Vert \xi_0 \right\Vert_{H^1(\mathbb{R}^6)}^2 + \left\Vert q \right\Vert_{H^1(\mathbb{R}^6)}^2 \right) + \frac{1}{2} \text{Tr} \left[ Q (A+B) Q \right] \nonumber \\
&\hspace{1.0cm} + \frac{2}{5 \beta} \text{Tr} \left( \alpha \overline{\alpha} - \alpha_0^w \overline{\alpha_0^w} \right)^2 + 2 \text{Re} \int_{\mathbb{R}^6} V \left( \frac{x-y}{h} \right) \Lambda(x,y)  \overline{\tilde{\alpha}_0^w}(x,y) \text{d}(x,y) - C_2 h^{3/2}, \label{eq:step1_29}
\end{align}
which holds for $h$ small enough. The reason why we keep both, $\xi$ and $\xi_0$ is due to technical reasons and will become apparent in step~2 and step~3. This ends the first step.
\subsection{Step 2: Bounds on the remaining non-positive terms} 
\label{subsec:lower_bound_step2}
In step~2 we derive bounds for the remaining non-positive terms, that is, the fifth term in the first line on the right-hand side of Eq.~\eqref{eq:step1_29} and the second term in the second line of the same equation.

Let us define the operator $\tilde{A}$ by
\begin{equation}
\tilde{A} = h^2 \sum_{n=1}^{\infty} 2 c^2 n^2 \frac{1}{E^2 +c^2 n^2} \left( \omega H_0 + H_0 \omega \right) \frac{1}{E^2 + c^2 n^2} \label{eq:step2_1}
\end{equation}
The matrix elements of $\tilde{A}$ will be denoted with $\tilde{a}_{ij}$ for $i,j=1,2$. In Fourier space, their integral kernels are given by
\begin{align}
\tilde{a}_{11}(p,q) &= h^2 \hat{W}(p-q) \left[ k(hp) + k(hq) \right] \zeta(p,q), \label{eq:step2_3} \\
\tilde{a}_{12}(p,q) &= h^2 \hat{W}(p-q) \left[ \hat{\Delta}_0(hp) - \hat{\Delta}_0(h q) \right] \zeta(p,q), \nonumber
\end{align}
where
\begin{equation}
\zeta(p,q) = \sum_{n=1}^{\infty} \frac{2 c^2 n^2}{E(hp)^2 + c^2 n^2} \frac{1}{E(hq)^2 + c^2 n^2}. \label{eq:step2_3b}
\end{equation}
Since $\omega^2$ is a positive operator we have 
\begin{equation}
\text{Tr} \left[ QAQ \right] \geq \text{Tr} \left[ Q \tilde{A} Q \right] = 2 \text{Tr} \left[ q \tilde{a}_{11} q \right] + 2 \text{Tr} \left[ \overline{\Lambda} \tilde{a}_{11} \Lambda \right] + 4 \text{Re} \text{Tr} \left[ \Lambda \tilde{a}_{12} q \right]. \label{eq:step2_2}
\end{equation}
The following Lemma summarizes some properties of the matrix elements of $\tilde{A}$ that we will need for our estimates.
\begin{lemma}
\label{lem:properties_of_zeta}
The operators $\tilde{a}_{11}$ and $\tilde{a}_{12}$ are bounded from $L^2(\mathbb{R}^3)$ to $L^2(\mathbb{R}^3)$ and obey the estimates
\begin{align}
\left\Vert (1+x^2) \tilde{a}_{11} (1+x^2) \right\Vert_{\infty} &\lesssim h^2, \\
\left\Vert (1+x^2) \tilde{a}_{12} (1+x^2) \right\Vert_{\infty} &\lesssim h^3. \nonumber
\end{align}
By $\left\Vert \cdot \right\Vert_{\infty}$ we denote the operator-norm of $\mathcal{L}(L^2(\mathbb{R}^3))$.
\end{lemma}
\begin{proof} 
Let the function $f$ be given by $f(x) = \frac{x^2}{(a^2 + x^2)(b^2 + x^2)}$. Its Fourier transform can be computed explicitly and reads \cite[p.~448]{Gradshteyn}
\begin{equation}
\hat{f}(k) = \frac{\pi}{a^2-b^2} \left( a e^{-a 2\pi |k|} - b e^{-b 2\pi |k|} \right). 
\label{eq:step2_5}
\end{equation}
Using the Poisson summation formula, see e.g. \cite{Stein_Fourieranalysis_Eucl_Spaces}, we find
\begin{equation}
\sum_{n=1}^{\infty} \frac{n^2}{(a^2 + n^2)(b^2 + n^2)} = -\frac{\pi}{2(a+b)} + \frac{\pi}{a^2-b^2} \left( \frac{a}{1-e^{-2\pi a}} - \frac{b}{1-e^{-2\pi b}} \right).
\label{eq:step2_6}
\end{equation}
Hence, the kernel $\zeta(p,q)$ can be written as
\begin{align}
\zeta(p,q) &= \frac{1}{E(p)+E(q)} \bigg\lbrace -\frac{\pi}{c^2} + \frac{2 \pi/c}{E(p)/c - E(q)/c}  \label{eq:step2_4} \\
&\hspace{6cm} \times \bigg[ \frac{E(p)/c}{1-e^{-2 \pi E(p)/c}} - \frac{E(q)/c}{1-e^{-2 \pi E(q)/c}} \bigg] \bigg\rbrace. \nonumber
\end{align}
Next, we will show that $\left\Vert (1+p^2) \zeta(p,q) \right\Vert_{L^{\infty}(\mathbb{R}^6)} < \infty$ holds. It is not hard to see that the $L^{\infty}(\mathbb{R}^6)$-norm of the term in the curly brackets on the right hand side of Eq.~\eqref{eq:step2_4} can be bounded by
\begin{equation}
\frac{\pi}{c} + 4 \pi + 2 \pi \left\Vert e^{-\beta E(p)} E(p) \frac{e^{-\beta (E(q) - E(p))} - 1}{E(q)-E(p)} \right\Vert_{L^{\infty}(\mathbb{R}^6)} < \infty. \label{eq:step2_8}
\end{equation}
Let us split $\left\Vert (1+p^2) \zeta(p,q) \right\Vert_{L^{\infty}(\mathbb{R}^6)}$ into two parts. If $p^2 \leq 2 \mu$ the function $E(p)$ may have zeros if $\hat{\Delta}_0(p) = 0$ and $p^2 = \mu$. On this set, we use $1+p^2 \leq 1 + 2 \mu$, which together with the series representation of $\zeta(p,q)$, yields the desired estimate. On the other hand, for $p^2 > 2\mu$ we can use Eq.~\eqref{eq:step2_8} and the obvious estimate 
\begin{equation}
\left\Vert \chi_{ B^c_{\sqrt{2\mu}} }(p) \frac{1+p^2}{E(p) + E(q)} \right\Vert_{L^{\infty}(\mathbb{R}^6)} < \infty, \label{eq:step2_9}
\end{equation}
where $\chi_{ B^c_{\sqrt{2\mu}} }(p)$ denotes the characteristic function of the complement of the ball with radius $\sqrt{2 \mu}$, to prove the claim. We also claim that $\left\Vert \zeta(p,q) \right\Vert_{W^{4,\infty}(\mathbb{R}^6)} < \infty$ holds. This can easily be seen by going back to the definition of $\zeta(p,q)$ as an infinite sum. Our assumption on $V$ imply that $| \partial_{p_i}^m (E(p)^2 + c^2 n^2)^{-1} | \lesssim (E(p)^2 + c^2 n^2)^{-1}$ for $i=1,2,3$ and $1 \leq m \leq 4$ which is enough to prove the claim. 

To keep our proof short, we only estimate the operator norm of the operator defined by the kernel $\tilde{a}_{11}(p,q) = \hat{W}(p-q) \left[ k(hp) + k(hq) \right] \zeta(hp,hq)$. The proof with the additional $(1+x^2)$-factors goes along the same lines.
\begin{align}
&\left( \int_{\mathbb{R}^3} \left| \int_{\mathbb{R}^3} h^2 \hat{W}(p-q) \left[ k(hp) + k(hq) \right] \zeta(hp,hq) \Psi(q) \text{d}q \right|^2 \right)^{1/2} \label{eq:step2_10} \\
&\hspace{0.5cm} \leq \left\Vert \left[ k(hp) + k(hq) \right] \zeta(hp,hq) \right\Vert_{L^{\infty}(\mathbb{R}^6)} \left( \int_{\mathbb{R}^3} \left( \int_{\mathbb{R}^3} \left| h^2 \hat{W}(p-q) \Psi(q) \right| \text{d}q \right)^2 \text{d}p \right)^{1/2} \nonumber \\
&\hspace{0.5cm}  \lesssim h^2  \left\Vert \Psi \right\Vert_{L^2(\mathbb{R}^3)}. \nonumber
\end{align}
The boundedness of $\left\Vert \left[ k(hp) + k(hq) \right] \zeta(hp,hq) \right\Vert_{L^{\infty}(\mathbb{R}^6)}$ is assured by the identity $\zeta(p,q)=\zeta(q,p)$ and the boundedness of $\left\Vert (1+p^2) \zeta(p,q) \right\Vert_{L^{\infty}(\mathbb{R}^6)}$.
\end{proof}
Using Lemma~\ref{lem:properties_of_zeta}, we will now estimate the terms on the right hand side of Eq.~\eqref{eq:step2_2}. Let us start with the first one. We have
\begin{equation}
\left| \text{Tr} \left[ q \tilde{a}_{11} q \right] \right| \leq \left\Vert \tilde{a}_{11} \right\Vert_{\infty} \left\Vert q \right\Vert_2^2 \lesssim h^2 \left\Vert q \right\Vert_2^2. \label{eq:step2_13} \\
\end{equation}
Here and in the following, we denote by $\left\Vert T \right\Vert_p = \left( \text{Tr} \left( T^{*} T \right)^{p/2} \right)^{1/p}$, $p \geq 1$, the $p$-th Schatten class-norm of the operator $T$. In order to estimate the second term on the right-hand side of Eq.~\eqref{eq:step2_2}, we use the decomposition $\alpha = \alpha_0 \psi + \xi_0$ and write 
\begin{align}
\text{Tr} \left[ \overline{\Lambda} \tilde{a}_{11} \Lambda \right] &= \text{Tr} \left[ \overline{\alpha_0 \varphi} \tilde{a}_{11} \alpha_0 \varphi \right] + 2 \text{Re} \text{Tr} \left[ \overline{\alpha_0 \varphi} \tilde{a}_{11} \left( \xi_0 - \tilde{\alpha}_0^w \right) \right] \label{eq:step2_14} \\
&\hspace{3.25cm} + \text{Tr} \left[ \left( \xi_0 - \tilde{\alpha}_0^w \right) \tilde{a}_{11} \left( \xi_0 - \tilde{\alpha}_0^w \right) \right], \nonumber
\end{align}
where $\varphi(x) = \psi(x) - 1$. We use the non-symmetric decomposition of $\alpha$ because like this we have less terms to estimate.
The last term on the right-hand side of Eq.~\eqref{eq:step2_14} can be estimated like the one in Eq.~\eqref{eq:step2_13}, which together with Lemma~\ref{lem:properties_of_alpha0_W} gives
\begin{equation}
\left| \text{Tr} \left[ \left( \xi_0^* + \overline{\tilde{\alpha}}_0^{w} \right) \tilde{a}_{11} \left( \xi_0 + \tilde{\alpha}_0^w \right) \right] \right| \lesssim h^2 \left( \left\Vert \xi_0 \right\Vert_{L^2(\mathbb{R}^6)}^2 + h^{3/2} \right). \label{eq:step2_15} 
\end{equation}
The terms with $\alpha_0 \varphi$ have to be estimated differently because we will not be able to control the $L^2(\mathbb{R}^3)$-norm of $\varphi$. Using the positive terms in Eq.~\eqref{eq:step1_29}, we will be able to dominate a term of the form $\int_{\mathbb{R}^3} \left( \left| \psi(x) \right|^2 - 1 \right) g(x) \text{d}x$ where $g(x)$ is a reasonably localized function. Since we will encounter expressions like this frequently in the following we introduce the notation $\Phi(x) = \left| \psi(x) \right|^2 - 1$. To bring the second term on the right-hand side of Eq.~\eqref{eq:step2_14} in this form, we compute
\begin{align}
&\left| 2 \text{Re} \text{Tr} \left[ \overline{\varphi} \alpha_0 \tilde{a}_{11} \left( \xi_0 - \tilde{\alpha}_0^w \right) \right] \right| \leq 2 \left| \text{Tr} \left[ \left( \frac{\overline{\varphi(x)}}{1+x^2} \right) \hat{\alpha}_0(-ih \nabla) \left( 1+ x^2 \right) \tilde{a}_{11} \left( \xi_0 - \tilde{\alpha}_0^w \right) \right] \right| \label{eq:step2_16} \\
&\hspace{5.4cm} + 2 \left| \text{Tr} \left[ \left( \frac{\overline{\varphi(x)}}{1+x^2} \right) \left[x^2, \hat{\alpha}_0(-ih \nabla) \right] \tilde{a}_{11} \left( \xi_0 - \tilde{\alpha}_0^w \right) \right] \right|. \nonumber
\end{align}
Since all other contributions can be treated similarly we only estimate the first of the two terms in the above equation. We have
\begin{align}
&\left| \text{Tr} \left[ \left( \frac{\overline{\varphi(x)}}{1+x^2} \right) \hat{\alpha}_0(-ih \nabla) \left( 1+ x^2 \right) \tilde{a}_{11} \left( \xi_0 - \tilde{\alpha}_0^w \right) \right] \right| \\
&\hspace{5cm} \leq \left\Vert \frac{\overline{\varphi(x)}}{1+x^2} \hat{\alpha}_0(-ih \nabla) \right\Vert_2 \left\Vert (1+x^2) \tilde{a}_{11} \right\Vert_{\infty} \left( \left\Vert \xi_0 \right\Vert_2 + \left\Vert \tilde{\alpha}_0^w \right\Vert_2 \right) \nonumber
\end{align}
To estimate the term proportional to $\varphi$, we use the Seiler-Simon inequality, see e.g. \cite[Theorem~4.1]{Simon_trace_ideals}. Together with the estimate $\left\Vert \hat{\alpha}_0 \right\Vert_{H^2(\mathbb{R}^3)} \lesssim h$, which follows from \cite[Proposition~5.6]{Frank_Lemm} and Lemma~\ref{lem:properties_of_alpha_0_1}, this gives
\begin{equation}
\left| 2 \text{Re} \text{Tr} \left[ \overline{\varphi} \alpha_0 \tilde{a}_{11} \left( \xi_0 - \tilde{\alpha}_0^w \right) \right] \right| \lesssim h^{3/2} \left( \int_{\mathbb{R}^3} \frac{|\varphi(x)|^2}{1+x^4} \text{d}x \right)^{1/2}  \left( \left\Vert \xi_0 \right\Vert_2 + \left\Vert \tilde{\alpha}_0^w \right\Vert_2 \right). \label{eq:step2_17} 
\end{equation}
We are not yet done, but close as the following estimate shows:
\begin{equation}
\int_{\mathbb{R}^3} \frac{\left| \varphi(x) \right|^2}{1+x^4}\text{d}x \leq 2 \int_{\mathbb{R}^3} \frac{\Phi(x)}{1+x^4}\text{d}x + 2 \left( \int_{\mathbb{R}^3} \frac{1}{1+x^4} \text{d}x \right) + 2 \left( \int_{\mathbb{R}^3} \frac{1}{1+x^4} \text{d}x \right)^{1/2}. \label{eq:step2_18}
\end{equation}
Insertion of Eq.~\eqref{eq:step2_18} into Eq.~\eqref{eq:step2_17} together with an application of Lemma~\ref{lem:properties_of_alpha0_W} yields
\begin{align}
\left| 2 \text{Re} \text{Tr} \overline{\varphi} \alpha_0 \tilde{a}_{11} \left( \xi_0 - \tilde{\alpha}_0^w \right) \right|  \lesssim h \int_{\mathbb{R}^3} \left| \hat{\Phi}(p) \widehat{\frac{1}{1+(\cdot)^4}} (p) \right| \text{d}p + h^2 \left\Vert \xi_0 \right\Vert_2^2 + h, \label{eq:step2_19}
\end{align}
where by $\widehat{\frac{1}{1+(\cdot)^4}} (p)$ we denote the Fourier transform of the function $x \mapsto (1+x^4)^{-1}$. All remaining terms can be estimated using the above ideas. We obtain
\begin{equation}
\left| \text{Tr} \left[ Q \tilde{A} Q \right] \right|  \lesssim h \int_{\mathbb{R}^3} \left| \hat{\Phi}(p) \widehat{ \frac{1}{1+(\cdot)^4} }(p) \right| \text{d}p + h^2 \left( \left\Vert q \right\Vert_2^2 + \left\Vert \xi_0 \right\Vert_2^2 \right) + h. \label{eq:step2_23}
\end{equation}
The estimates for $\text{Tr} \left[ Q B Q \right]$ go along the same lines and we leave them to the reader. We find 
\begin{align}
\left| \text{Tr} Q B Q \right| \lesssim h^3 \left[ \left( \int_{\mathbb{R}^3} \left| \hat{\Phi}(p) \widehat{\frac{1}{1+(\cdot)^4}}(p) \right| \text{d}p \right) + 1 \right] + h^4 \left( \left\Vert q \right\Vert_2^2 +  \left\Vert \xi_0 \right\Vert_2^2  \right) + h^{3}.
\label{eq:step2_29}
\end{align}
This ends the construction of a bound for the term $\text{Tr} \left[ Q(A+B)Q \right]$.

The remaining non-positive term to estimate is the last term on the right-hand side of Eq.~\eqref{eq:step1_29}. We use the decomposition $\alpha=\alpha_0 \psi + \xi_0$ and only consider the term proportional to $\alpha_0 \varphi$.
\begin{align}
&\Bigg| \int_{\mathbb{R}^6} V\left( \frac{x-y}{h} \right) h^{-3} \alpha_0\left( \frac{x-y}{h} \right) \varphi(y) \overline{\tilde{\alpha}}_0^w(x,y) \text{d}(x,y) \Bigg| =  \left| \text{Tr} \left[ \Delta_0 \frac{\varphi}{1+x^2} (1+x^2) \overline{ \tilde{\alpha}}_0^w  \right] \right|  \label{eq:step2_30} \\
&\hspace{3cm} \lesssim h^{-1/2} \left[ \left( \int_{\mathbb{R}^3} \left| \hat{\Phi}(p) \widehat{\frac{1}{1+(\cdot)^4}}(p) \right| \text{d}p \right)^{1/2} + 1 \right] \left\Vert (1+x^2) \tilde{\alpha}_0^w \right\Vert_2 \nonumber \\
&\hspace{3cm} \lesssim h \int_{\mathbb{R}^3} \left| \hat{\Phi}(p) \widehat{\frac{1}{1+(\cdot)^4}}(p) \right| \text{d}p + h. \nonumber
\end{align}
We used Eq.~\eqref{eq:step2_18} to come to the second line and Lemma~\ref{lem:properties_of_alpha0_W_b} to come to the last line.

When we put our estimates for the non-positive terms in Eq.~\eqref{eq:step1_29} together, that is, Eq.~\eqref{eq:step2_29} and Eq.~\eqref{eq:step2_30}, we arrive at the following lower bound for the BCS functional: 
\begin{align}
\mathcal{F}_{\beta}(\Gamma,\Gamma_0^w) &\geq C_1 \left( h \left\Vert \nabla \psi \right\Vert_{L^2(\mathbb{R}^3)}^2 +  \left\Vert \xi \right\Vert_{H_1(\mathbb{R}^6)}^2 + \left\Vert \xi_0 \right\Vert_{H_1(\mathbb{R}^6)}^2 + \left\Vert q \right\Vert_{H^1(\mathbb{R}^6)}^2  \right) \label{eq:step2_31} \\
&\ \ + \frac{2}{5 \beta} \text{Tr} \left( \alpha \overline{\alpha} - \alpha_0^w \overline{\alpha}_0^w \right)^2 - h C_2 \int_{\mathbb{R}^3} \left| \hat{\Phi}(p) \widehat{ \frac{1}{1+(\cdot)^4} }(p) \right| \text{d}p - C_3 h. \nonumber
\end{align}
Eq.~\eqref{eq:step2_31} holds for appropriately chosen constants $C_1, C_2, C_3 > 0$ and $h$ small enough. This ends the second step.
\subsection{Step 3: Construction of the lower bound and a-priori estimates} 
\label{subsec:lower_bound_step3}
The construction of the lower bound starting from Eq.~\eqref{eq:step2_31} needs one crucial ingredient - estimates of the form $\left\Vert \nabla \psi \right\Vert_{L^2(\mathbb{R}^3)}^2 \lesssim [ 1+ \Vert \hat{\Phi} \Vert_{L^2(B_r)} ]$ and $\left\Vert \xi \right\Vert_{L^2(\mathbb{R}^6)}^2 \lesssim h [ 1+ \Vert \hat{\Phi} \Vert_{L^2(B_r)} ]$ for some appropriately chosen $r > 1$. By $B_r$ we denote the ball of radius $r$ centered around zero and as above $\Phi = | \psi |^2-1$. We will see that the condition $\mathcal{F}_{\beta}(\Gamma,\Gamma_0^w) \leq 0$ is strong enough to guarantee such bounds. The estimates imply a separation of scales with respect to the decomposition $\alpha = ( \alpha_0 \psi + \psi \alpha_0 )/2 + \xi$. On the other hand, they allow us to show that the part of $-h C_2 \int_{\mathbb{R}^3} \left| \hat{\Phi}(p) \widehat{ \frac{1}{1+(\cdot)^4} }(p) \right| \text{d}p$ containing the high Fourier modes of $\hat{\Phi}$ can be controlled by $\left\Vert \nabla \psi \right\Vert_{L^2(\mathbb{R}^3)}^2$ and that the contribution coming from the low Fourier modes is dominated by $\frac{2}{5 \beta} \text{Tr} \left( \alpha \overline{\alpha} - \alpha_0^w \overline{\alpha}_0^w \right)^2$. We start with the following technical lemma which will allow us to obtain the estimates we are heading for:
\begin{lemma}
Let $g : \mathbb{R}^3 \to \mathbb{R}_+$ be a measurable function, $C_1 > 0$ and assume that for some $\beta \in (2,\infty)$ one has $\left\Vert (1+|p|^{\beta}) g \right\Vert_{L^1(\mathbb{R}^3)} + \left\Vert (1+|p|^{\beta}) g \right\Vert_{L^{\infty}(\mathbb{R}^{3})} < \infty$. Then there exists a number $R>0$ such that for all $r \geq R$ one has
\begin{equation}
\int_{\mathbb{R}^3} \left| \nabla \psi(x) \right|^2 \text{d}x - C_1 \int_{B_{r}^c} \left| \hat{\Phi}(p) \right| g(p) \text{d}p \geq  - C(r),
\label{eq:step3_1}
\end{equation}
where $\Phi = \left| \psi \right|^2 - 1$. The constant $C(r)>0$ satisfies $\lim_{r \rightarrow \infty} C(r) = 0$.
\label{lem:step3_lem1}
\end{lemma}
\begin{proof} We start by expressing the above quantities in terms of $\widehat{|\psi|-1}$, the Fourier transform of $|\psi|-1$. This can be done because $(| \psi |^2 - 1) = (|\psi|-1)^2 + 2(|\psi|-1)$. An application of the triangle inequality yields
\begin{align}
\int_{B_r^c} \left| \hat{\Phi}(p) \right| g(p) \text{d}p &\leq \int_{B_r^c} \left| \widehat{|\psi|-1} \ast \widehat{|\psi|-1 } (p)\right| g(p) \text{d}p + 2 \int_{B_r^c} \left| \widehat{|\psi|-1} (p)\right| g(p) \text{d}p \nonumber \\
&\leq \frac{1}{1+r^{\beta}} \int_{B_r^c} \left| \widehat{(|\psi|-1)} \ast \widehat{(|\psi|-1)}(p) \left[ \left( 1+|p|^{\beta} \right) \right] g(p) \right| \text{d}p \nonumber \\
&\hspace{3cm} + 2 \left\Vert \widehat{(|\psi|-1)} \right\Vert_{L^2(B_r^c)} \left\Vert g \right\Vert_{L^2(B_r^c)}. \label{eq:step3_2}
\end{align}
To obtain a bound for the first two terms on the right-hand side of Eq.~\eqref{eq:step3_2}, we write all functions as a sum of one part living in $B_r$ and another one living in its complement. In other words, we insert $1 = \chi_{B_r}(p) + \chi_{B_r^c}(p)$ in front of each function where $\chi_{B_r}$ denotes the characteristic function of the ball with radius $r$ centered around zero. We note that $\Vert \widehat{(|\psi|-1)} \Vert_{L^1(B_r)} \leq \Vert \frac{1}{|p|} \Vert_{L^2(B_r)} \Vert p \ \widehat{(|\psi|-1)} \Vert_{L^2(B_r)}$, and therefore an application of Young's inequality tells us that 
\begin{align}
&\int_{B_r^c} \big| \widehat{(|\psi|-1)} \ast \widehat{(|\psi|-1)}(p) \left[ \left( 1+|p|^{\beta} \right) \right]  \hat{g}(p) \big| \text{d}p \label{eq:step3_3} \\
& \hspace{1.5cm} \lesssim \left\Vert \widehat{|\psi|-1} \right\Vert_{L^2(B_r^c)}^2 \left\Vert \left(1+|p|^{\beta}\right) \hat{g} \right\Vert_{L^1(B_r^c)} \nonumber \\
& \hspace{2cm} + \left\Vert \widehat{p (|\psi|-1)} \right\Vert_{L^2(B_r)}^2 \left\Vert \frac{1}{|p|} \right\Vert_{L^2(B_r)}^2 \left\Vert (1+|p|^{\beta}) \hat{g} \right\Vert_{L^{\infty}(B_r^c)} \nonumber \\
& \hspace{2cm} + \left\Vert \widehat{(|\psi|-1)} \right\Vert_{L^2(B_r^c)} \left\Vert \widehat{p \ (|\psi|-1)} \right\Vert_{L^2(B_r)} \left\Vert \frac{1}{|p|} \right\Vert_{L^2(B_r)} \left\Vert (1+|p|^{\beta}) \hat{g} \right\Vert_{L^2(B_r^c)}. \nonumber
\end{align} 
The gradient term on the other hand is bounded from below by
\begin{align}
\left\Vert \nabla \psi \right\Vert_{L^2(\mathbb{R}^3)}^2 &\geq \left\Vert \nabla (|\psi(x)|-1) \right\Vert_{L^2(\mathbb{R}^3)}^2 \label{eq:step3_4} \\
&\geq r^2 \left\Vert \widehat{(|\psi|-1)} \right\Vert_{L^2(B_r(0)^c))}^2 + \left\Vert \widehat{p (|\psi|-1)} \right\Vert_{L^2(B_r(0))}^2. \nonumber
\end{align}  
Since $\int_{B(r)} \frac{1}{|p|} \ \text{d}p = 2 \pi r^2$, we finally obtain
\begin{align}
\int_{\mathbb{R}^3} &\left| \nabla \psi(x) \right|^2 \text{d}x - C_1 \int_{B_r^c} \left|\hat{\Phi}(p) \right| g(p) \text{d}p \label{eq:step3_5} \\
&\hspace{2cm} \geq r^2 \left\Vert \widehat{(|\psi|-1)} \right\Vert_{L^2(B_r^c))}^2 + \left\Vert \widehat{p \ (|\psi|-1)} \right\Vert_{L^2(B_r))}^2 \nonumber \\
&\hspace{3.5cm} - 2 C_1 \left\Vert \widehat{(|\psi|-1)} \right\Vert_{L^2(B_r^c)} \left\Vert g \right\Vert_{L^2\left(B_r^c\right)} \nonumber \\
&\hspace{3.5cm} - \frac{C_1 C_2 }{1+|r|^{\beta}} \bigg[ \left\Vert \widehat{(|\psi|-1)} \right\Vert_{L^2(B_r^c)}^2 \left\Vert (1+|p|^{\beta}) g \right\Vert_{L^1(B_r^c)} \nonumber \\
&\hspace{3.5cm} + \left\Vert \widehat{p (|\psi|-1)} \right\Vert_{L^2(B_r)}^2 r^2  \left\Vert (1+|p|^{\beta}) g \right\Vert_{L^{\infty}(B_r^c)} \nonumber \\
&\hspace{3.5cm} + \left\Vert \widehat{(|\psi|-1)} \right\Vert_{L^2(B_r^c)} \left\Vert \widehat{p (|\psi|-1)} \right\Vert_{L^2(B_r)} r \left\Vert (1+|p|^{\beta}) g \right\Vert_{L^2(B_r^c)} \bigg] \nonumber
\end{align}
for an appropriately chosen constant $C_2 > 0$. Choosing $\beta > 2$ and $r$ large enough, we easily see that the expression on the right-hand side behaves as claimed. 
\end{proof}
Choose $\Gamma$ such that $\mathcal{F}_{\beta}(\Gamma,\Gamma_0^w) \leq 0$. This is always possible because $\mathcal{F}_{\beta}(\Gamma_0^w, \Gamma_0^w) = 0$. An application of Lemma~\ref{lem:step3_lem1} in Eq.~\eqref{eq:step2_31} gives
\begin{equation}
0 \geq C_1 \left( h \left\Vert \nabla \psi \right\Vert_{L^2(\mathbb{R}^3)}^2 + \left\Vert \xi \right\Vert_{L^2(\mathbb{R}^6)}^2 \right) - h C_2 \int_{B_r} \left| \hat{\Phi}(p) \right| \text{d}p - C_2 h \label{eq:step3_6}
\end{equation}
for $r$ large enough and with appropriately chosen constants $C_1, C_2$ independent of $r$ and $h$. Eq.~\eqref{eq:step3_6} yields the bounds
\begin{align}
\left\Vert \xi \right\Vert_{L^2(\mathbb{R}^6)}^2 + \left\Vert \xi_0 \right\Vert_{L^2(\mathbb{R}^6)}^2 &\lesssim h \left( \left\Vert \hat{\Phi} \right\Vert_{L^2(B_r)} + 1 \right), \label{eq:step3_7} \\
\left\Vert \nabla \psi \right\Vert_{L^2(\mathbb{R}^3)}^2 &\lesssim \left\Vert \hat{\Phi} \right\Vert_{L^2(B_r)} + 1. \nonumber
\end{align}
We continue with Eq.~\eqref{eq:step2_31} and apply Lemma~\ref{lem:step3_lem1} another time which gives
\begin{align}
\mathcal{F}_{\beta}(\Gamma,\Gamma_0^w) &\geq C_1 \left( h \left\Vert \nabla \psi \right\Vert_{L^2(\mathbb{R}^3)}^2 + \left\Vert \xi \right\Vert_{H_1(\mathbb{R}^6)}^2 + \left\Vert \xi_0 \right\Vert_{L^2(\mathbb{R}^6)}^2 + \left\Vert q \right\Vert_{H^1(\mathbb{R}^6)}^2 \right) + \frac{2}{5 \beta} \text{Tr} \left( \alpha \overline{\alpha} - \alpha_0^w \overline{\alpha}_0^w \right)^2 \nonumber \\
&\hspace{5cm} - C_2 h \int_{B_r} \left| \hat{\Phi}(p) \right|\text{d}p - C_2 h, \label{eq:step3_9}
\end{align} 
for appropriately chosen constants $C_1,C_2 > 0$ independ of $r$ and $h$. We define the functions $\Phi_{\leq}$ and $\Phi_{>}$ to be the inverse Fourier transforms of $\hat{\Phi}(p) \chi_{ \left\lbrace |x| \leq r \right\rbrace }(p)$ and $\hat{\Phi}(p) \chi_{ \left\lbrace |x|> r \right\rbrace }(p)$, respectively. When we insert $\alpha = \alpha_0 \psi + \xi_0$ into the last term on the right-hand side of Eq.~\eqref{eq:step3_9} and use $\text{Tr} \left[ \alpha_0 \Phi_{\leq} \alpha_0^2 \Phi_{>} \alpha_0 \right] = 0$, which results from $\hat{\Phi}_{\leq}(p)$ and $\hat{\Phi}_{>}(p)$ having disjoint support, we obtain
\begin{align}
&\text{Tr} \left[ \left( \alpha \overline{\alpha} - \alpha_0^w \overline{\alpha}_0^w \right)^2 \right]  \label{eq:step3_10} \\
&\hspace{1.5cm}\geq \text{Tr} \left[ \left(\alpha_0 \Phi_{\leq} \alpha_0 \right)^2 \right] + 2 \text{Re} \text{Tr} \left[ \alpha_0 \Phi_{\leq} \alpha_0 \left( \xi_0 \xi_0^* + \alpha_0 \psi \xi_0^* + \xi_0 \overline{\psi} \alpha_0 -\alpha_0^w \overline{\alpha_0^w} + \alpha_0^2 \right) \right]. \nonumber
\end{align}
The first term on the right-hand side of Eq.~\eqref{eq:step3_10} can be estimated by
\begin{align}
\text{Tr} \left[ \left(\alpha_0 \Phi_{\leq} \alpha_0 \right)^2 \right] &= h^{-3} \int_{\mathbb{R}^3} \left| \hat{\Phi}_{\leq}(p) \right|^2 \hat{\alpha}_0^2 \ast \hat{\alpha}_0^2(hp) \text{d}p \label{eq:step3_11} \\
&\geq \inf_{p \in B_r}\left( \hat{\alpha}_0^2 \ast \hat{\alpha}_0^2(hp) \right) h^{-3} \int_{B_r} \left| \hat{\Phi}_{\leq}(p) \right|^2 \text{d}p. \nonumber
\end{align} 
As the unique ground state of the real operator $K_T^{\Delta_0} + V$,  $\alpha_0$ as well as its Fourier transform are real functions. Since $\hat{\alpha}_0^2 \ast \hat{\alpha}_0^2(0) = \int_{\mathbb{R}^3} \hat{\alpha}_0(p)^4 \text{d}p = c > 0$, it can easily be seen with the help of \cite[Theorem~2.10]{Frank_Lemm} that $\inf_{p \in B_r}\left( \hat{\alpha}_0^2 \ast \hat{\alpha}_0^2(hp) \right) \geq h^4 c/2$ for $hr$ small enough. 

With the help of Eq.~\eqref{eq:step3_7} we can easily estimate the remaining terms on the right-hand side of Eq.~\eqref{eq:step3_10}. Using the fact that lower Schatten-class norms dominate higher ones, we find
\begin{align}
\left| \text{Tr} \left[ \alpha_0 \Phi_{\leq} \alpha_0 \xi_0 \xi_0^* \right] \right| &\leq \left\Vert \alpha_0 \Phi_{\leq} \alpha_0 \right\Vert_2 \left\Vert \xi_0 \right\Vert_4^2 \label{eq:step3_12} \\
&\lesssim h^{3/2} \left\Vert \hat{\Phi}_{\leq} \right\Vert_{L^2(\mathbb{R}^3)}^2 + h^{3/2}. \nonumber
\end{align}
With Young's inequality, one similarly finds 
\begin{align}
\left| \text{Tr} \left[ \alpha_0 \Phi_{\leq} \alpha_0 \left( \alpha_0 \psi \xi_0^* + \xi_0 \overline{\psi} \alpha_0 \right) \right] \right| &\lesssim \left\Vert \alpha_0 \Phi_{\leq} \alpha_0 \right\Vert_2 \left\Vert \alpha_0 \psi \right\Vert_{\infty} \left\Vert \xi_0 \right\Vert_2 \label{eq:step3_13} \\
&\lesssim h^{3/2} \left( \left\Vert \hat{\Phi}_{\leq} \right\Vert_{L^2(\mathbb{R}^3)}^2 + 1 \right). \nonumber
\end{align}
On the other hand, an application of Lemma~\ref{lem:properties_of_alpha0_W} yields
\begin{align}
\left| \text{Tr} \left[ \alpha_0 \Phi_{\leq} \alpha_0 \left( -\alpha_0^w \overline{\alpha_0^w} + \alpha_0^2 \right) \right] \right| &\leq \left\Vert \alpha_0 \Phi_{\leq} \alpha_0 \right\Vert_2 \left( \left\Vert -\alpha_0^w \overline{\alpha_0^w} + \alpha_0 \overline{\alpha}_0^w \right\Vert_2 + \left\Vert - \alpha_0 \overline{\alpha}_0^w + \alpha_0^2 \right\Vert_2 \right) \nonumber \\
&\lesssim h^{2} \left\Vert \hat{\Phi}_{\leq} \right\Vert_{L^2(\mathbb{R}^3)}. \label{eq:step3_14}
\end{align}
When we insert the above findings into Eq.~\eqref{eq:step3_10} and afterwards insert the resulting expression into Eq.~\eqref{eq:step3_9}, we arrive at
\begin{align}
\mathcal{F}_{\beta}(\Gamma,\Gamma_0^w) &\geq C_1 \Bigg( h \left\Vert \nabla \psi \right\Vert_{L^2(\mathbb{R}^3)}^2 + h \left\Vert \widehat{|\psi|^2-1} \right\Vert_{L^2(B_r)}^2 + \left\Vert \xi \right\Vert_{H_1(\mathbb{R}^6)}^2  \label{eq:step3_15} \\
&\hspace{9cm} + \left\Vert q \right\Vert_{H^1(\mathbb{R}^6)}^2 \Bigg) - C_2 h. \nonumber
\end{align}
It remains to show that the $L^2(B_r)$-norm of $\left\vert \psi \right\vert^2 - 1$ can be replaced by the $L^2(\mathbb{R}^3)$-norm.

Let us introduce the function $\eta(x) = |\psi(x)|-1$. We have 
\begin{equation}
\Phi(x) = |\psi(x)|^2-1 = \eta(x)^2 - 2 \eta(x).  \label{eq:step3_16}
\end{equation}
Now we decompose $\eta$ into two parts, $\hat{\eta}(p) = \hat{\eta}_1(p) + \hat{\eta}_2(p)$ with $\hat{\eta}_1(p) = \hat{\eta}(p) \chi_{B_s}(p)$ and $\hat{\eta}_2(p) = \hat{\eta}(p) \chi_{B_s^c}(p)$ for some number $s>0$. They satisfy the following bounds 
\begin{align}
\left\Vert \hat{\eta}_1 \right\Vert_{L^1(\mathbb{R}^3)} &\lesssim s^{1/2} \left \Vert \nabla \psi \right\Vert_{L^2(\mathbb{R}^3)}, \label{eq:step3_17} \\
\left\Vert \hat{\eta}_2 \right\Vert_{L^2(\mathbb{R}^3)} &\lesssim s^{-1} \left \Vert \nabla \psi \right\Vert_{L^2(\mathbb{R}^3)}. \nonumber
\end{align}
We also have 
\begin{equation}
\left\Vert \eta_2 \right\Vert_{L^4(\mathbb{R}^3)} \leq \left\Vert \eta_2 \right\Vert_{L^2(\mathbb{R}^3)}^{1/4} \left\Vert \eta_2 \right\Vert_{L^6(\mathbb{R}^3)}^{3/4} \lesssim s^{-1/4} \left\Vert \nabla \psi \right\Vert_{L^2(\mathbb{R}^3)}. \label{eq:step3_17b}
\end{equation}
Next we choose $s = r/2$ and estimate $\hat{\Phi}$ outside $B_r$ using that $\hat{\eta}_1$ has compact support as well as Eq.~\eqref{eq:step3_17} and Eq.~\eqref{eq:step3_17b}:
\begin{align}
\left\Vert \hat{\Phi} \right\Vert_{L^2(B_r^c)} &= \left\Vert \hat{\eta}_1 \ast \hat{\eta}_1 + \hat{\eta}_2 \ast \hat{\eta}_2 + 2 \hat{\eta}_1 \ast \hat{\eta}_2 - 2\hat{\eta}_1 - \hat{\eta}_2 \right\Vert_{L^2(B_r^c)} \label{eq:step3_18} \\
&\lesssim s^{-1/2} \left\Vert \nabla \psi \right\Vert_{L^2(\mathbb{R}^3)}^2 + s^{-1} \left\Vert \nabla \psi \right\Vert_{L^2(\mathbb{R}^3)}. \nonumber
\end{align}
Going back to Eq.~\eqref{eq:step3_11}, we see that we can choose $r = h^{-1/2}$ without changing the rest of the argument until Eq.~\eqref{eq:step3_15}. In particular, the function $\psi$ is such that $\left\Vert \nabla \psi \right\Vert_{L^2(\mathbb{R}^3)}$ is still bounded independently of $h$. With Eq.~\eqref{eq:step3_18}, we find
\begin{equation}
\left\Vert \hat{\Phi} \right\Vert_{L^2(B_r^c)}^2 \lesssim h^{1/2}.
\end{equation}
This concludes the proof of Theorem~\ref{thm:lower_bound_BCS}.
\begin{appendix}
\section{Properties of $\alpha_0$ and $\alpha_0^w$}
\label{Appendix}
In this section we establish some properties of the minimizer of the translation-invariant BCS functional $\alpha_0$ and of $\alpha_0^w$. We remind that for any minimizing pair $(\gamma_0,\alpha_0)$ of the BCS functional, $\alpha_0$ is a pointwise a.e. solution of the BCS gap equation, see \cite{BCS_general_pair_interaction}.
\begin{lemma}
Let $V \in L^{3/2}(\mathbb{R}^3)$, $k \in \mathbb{N}_0$ and let $\alpha_0$ with $\alpha_0 \not\equiv 0$ be a pointwise a.e. solution of the BCS gap equation, Eq.~\eqref{eq:ch3setupA4}. Then $V \in W^{k,\infty}(\mathbb{R}^3)$ implies $\alpha_0 \in H^{k+2}(\mathbb{R}^3)$ and $\hat{V} \in H^k(\mathbb{R}^3)$ implies $\hat{\alpha}_0 \in H^k(\mathbb{R}^3)$.
\label{lem:properties_of_alpha_0_1}
\end{lemma}
\begin{proof}
The gap equation reads
\begin{equation}
K_T^{\Delta_0}(p) \hat{\alpha}_0(p) + (2 \pi)^{-3/2} \hat{V} \ast \hat{\alpha}_0(p) = 0.
\label{eq:appA_1}
\end{equation}
We can estimate the $L^2(\mathbb{R}^3)$-norm of derivatives of $\alpha_0$ as follows:
\begin{align}
\left\Vert \partial_i^m \alpha_0 \right\Vert_{L^2(\mathbb{R}^3)} &= (2\pi)^{-3/2} \left\Vert p_i^m \frac{1}{K_T^{\Delta_0}} \hat{V} \ast \hat{\alpha}_0 \right\Vert_{L^2(\mathbb{R}^3)} \leq \left\Vert \frac{p^2}{K_T^{\Delta_0}} \right\Vert_{L^{\infty}(\mathbb{R}^3)} \left\Vert p_i^{m-2} \hat{V} \ast \hat{\alpha}_0 \right\Vert_{L^2(\mathbb{R}^3)} \nonumber \\
&\lesssim \sum_{j=0}^{m-2} \left\Vert \partial_i^j V \right\Vert_{\infty} \left\Vert \partial_i^{m-2-j} \alpha_0 \right\Vert_{L^2(\mathbb{R}^3)}. \label{eq:appA_2}
\end{align}
To come to the last line, we used that $K_T^{\Delta_0} \lesssim (1+p^2)$. Eq.~\eqref{eq:appA_2} shows that $V \in W^{k,\infty}(\mathbb{R}^3)$ implies $\alpha_0 \in H^{k+2}(\mathbb{R}^3)$.

In order to obtain the second property, we again use Eq.~\eqref{eq:appA_1} together with the fact that in the second term, $\hat{\alpha}_0$ appears only in convolution with $\hat{V}$. Hence, derivatives act only on $\hat{V}$ and not on $\hat{\alpha}_0$. We compute
\begin{align}
\left\Vert \partial_i \hat{\alpha}_0(p) \right\Vert_{L^2(\mathbb{R}^3)} &\lesssim \left\Vert \frac{2 k(p) p_i +  \hat{\Delta}_0(p) \left( \partial_i \hat{\Delta}_0(p) \right) }{4 T E(p) \left( K_T^{\Delta_0}(p) \right)^2}  \left[ \frac{\frac{E(p)}{2T} - \frac{1}{2} \sinh\left( \frac{E(p)}{T} \right)}{E(p)^2 \cosh\left( \frac{E(p)}{2T} \right)^2 } \right] \hat{\Delta}_0(p) \right\Vert_{L^2(\mathbb{R}^3)}  \nonumber \\
&\hspace{0.5cm} + \left\Vert \frac{1}{K_T^{\Delta_0}(p)} \left( \partial_i \hat{V} \right)\ast \hat{\alpha}_0(p) \right\Vert_{L^2(\mathbb{R}^3)} \label{eq:appA_3} \\
&\leq C\left( \left\Vert \hat{V} \right\Vert_{H^1(\mathbb{R}^3)} \right) \left\Vert \hat{\alpha}_0 \right\Vert_{L^2(\mathbb{R}^3)}. \nonumber
\end{align}
To obtain the result, we used the identity $\hat{\Delta}_0(p) = -2 (2 \pi)^{-3/2} \hat{V} \ast \hat{\alpha}_0(p)$. Looking at the above terms, one can easily see that another differentiation does not change this structure, except that second derivatives of $\hat{V}$ appear. The extension to $k$ derivatives therefore is a simple exercise in differentiation.
\end{proof}
\begin{lemma}
Let $V \in H^1(\mathbb{R}^3) \cap W^{1,\infty}(\mathbb{R}^3)$, $\hat{V} \in L^1(\mathbb{R}^3)$ and $W \in H^1(\mathbb{R}^3) \cap W^{1,\infty}(\mathbb{R}^3)$. Then $\left\Vert \alpha_0^w(x,y) - h^{-3} \alpha_0 \left( \frac{x-y}{h} \right) \right\Vert_{H^1(\mathbb{R}^6)} \lesssim h^{3/2}$. 
\label{lem:properties_of_alpha0_W}
\end{lemma}
\begin{proof}
To prove the claim, we make use of a representation of the operator
\begin{equation}
\alpha_0^w = \left[ \left( 1+e^{\beta H_0^w} \right)^{-1} \right]_{12}
\label{eq:properties_of_a0_W_definition}
\end{equation}
in terms of a Cauchy integral. This is possible because the function $g(z) = \left( 1 + e^z \right)^{-1}$ is analytic in the strip $\left\lbrace z \in \mathbb{C} | \ | \text{Im}(z) | < \pi \right\rbrace$. Let  
\begin{equation}
\mathcal{C}_R = \left\lbrace r - i \pi/(2 \beta), r \in [-R,R] \right\rbrace \cup \left\lbrace -r + i \pi/(2 \beta), r \in [-R,R] \right\rbrace. \label{eq:properties_of_a0_W_1}
\end{equation}
Then
\begin{equation}
\alpha_0^w = \lim_{R \rightarrow \infty} \frac{1}{2 \pi i} \int_{\mathcal{C}_R} g(\beta z) \left[ \frac{1}{z-H_0^w} \right]_{12} \text{d}z,
\label{eq:properties_of_a0_W_2}
\end{equation}
where the limit $R \rightarrow \infty$ is to be taken in the weak operator topology. For further details on the construction of the above integral, see \cite[p.~696, p.~704]{GL_original}. In the following, we will often write $\lim_{R \rightarrow \infty}\int_{\mathcal{C}_R} = \int_{\mathcal{C}}$ with $\mathcal{C} = \cup_{R \geq 0} \ \mathcal{C}_R$ to denote the above limit. Let us write the resolvent of $H_0^w$ as
\begin{equation}
\frac{1}{z-H_0^w} = \frac{1}{z-H_0} + \frac{1}{z-H_0} h^2 \omega \frac{1}{z-H_0} + \frac{1}{z-H_0} h^2 \omega \frac{1}{z-H_0} h^2 \omega \frac{1}{z-H_0^w}. \label{eq:properties_of_a0_W_5}
\end{equation}
Since $\tilde{\alpha}_0^w(x,y) = \tilde{\alpha}_0^w(y,x)$, it is sufficient to consider a derivative acting on the $x$-component of this function. 

The resolvent of  $H_0$ can be computed explicitly and reads:
\begin{equation}
\frac{1}{z-H_0} = \frac{1}{(z-E)(z+E)} \begin{pmatrix} z+k & \Delta_0 \\ \Delta_0 & z-k \end{pmatrix}.
\label{eq:properties_of_a0_W_6}
\end{equation}
This implies
\begin{align}
&\left( \frac{1}{2 \pi i} \int_{\mathcal{C}} g(\beta z) \left[ \frac{1}{z-H_0} h^2 \omega \frac{1}{z-H_0} \right]_{12} \ \text{d}z \right)(p,q) \label{eq:properties_of_a0_W_8} \\
&\hspace{1cm} = \left[ \frac{\beta}{2} \frac{g_0(\beta E(hq)) - g_0(\beta E(hp))}{E(hp) - E(hq)} \right] \times \left[  h^2 \hat{W}(p-q) \frac{\hat{\Delta}(hq) k(hp) + \hat{\Delta}(hp) k(hq)}{E(hp) + E(hq)} \right], \nonumber
\end{align}
where $g_0(z) = \tanh(z/2)/z$. We have to bound the $L^2(\mathbb{R}^6)$-norm of $hp$ times this kernel, which can be done as follows: Take the $L^{\infty}(\mathbb{R}^6)$-norm of the term in the left bracket on the right hand side of the above equation. It can easily be seen to be bounded independently of $h$. On the other hand, the $L^2(\mathbb{R}^6)$-norm of $hp$ times the second bracket on the right hand side of Eq.~\eqref{eq:properties_of_a0_W_8} can be bounded by a constant times $h^{3/2}$.

To bound the contribution coming from the third term on the right-hand side of Eq.~\eqref{eq:properties_of_a0_W_5}, we first have to compute the upper right component of the operator-valued matrix under consideration.
\begin{align}
&\left[ \frac{1}{z-H_0} h^2 \omega \frac{1}{z-H_0} h^2 \omega \frac{1}{z-H_0^w} \right]_{12} \label{eq:properties_of_a0_W_11} \\
&\hspace{3cm} = \left[ D(z+k) h^2W D(z+k) - D \Delta_0 h^2 W D \Delta_0 \right] h^2 W \left[ \frac{1}{z-H_0^w} \right]_{12} \nonumber \\
&\hspace{3.5cm} - \left[ D(z+k) h^2 W D \Delta_0 - D \Delta_0 h^2 W D (z-k) \right] h^2 W \left[ \frac{1}{z-H_0^w} \right]_{22}. \nonumber 
\end{align}
For the sake of convenience, we have introduced the abbreviation $D = (z^2 - E^2)^{-1}$. We will only show how to bound the Cauchy integral of the first and the last terms on the right-hand side of Eq.~\eqref{eq:properties_of_a0_W_11}. All remaining terms can be bounded with similar arguments. Let us start with the contribution coming from the first term:
\begin{align}
&\left\Vert \left(-ih\nabla_j \right) D (z+k) h^2 W D(z+k) h^2W \left( \frac{1}{z-H_0^w} \right)_{12} \right\Vert_2 \label{eq:properties_of_a0_W_12} \\
&\hspace{4cm} \lesssim h^4 \left\Vert D(z+k) \right\Vert_{\infty}^2 \left\Vert W \right\Vert_{W^{1,\infty}(\mathbb{R}^3)} \left\Vert (-ih \nabla_j) W \left[ \frac{1}{z-H_0^w} \right]_{12} \right\Vert_2. \nonumber
\end{align}
The first factor on the right-hand side of Eq.~\eqref{eq:properties_of_a0_W_12} can be estimated by
\begin{equation}
\left\Vert D(z+k) \right\Vert_{\infty} \lesssim \begin{cases} 1 & \text{ for } r \gg 1, \\ \frac{1}{1+|r|} & \text{ for } r \ll -1 \end{cases} \label{eq:properties_of_a0_W_13}
\end{equation}
and has to be understood to hold on the contour $\mathcal{C}$. By $r$ we refer to the natural coordinates on $\mathcal{C}$. Let us note that $g(\beta z)$ decays exponentially for $r \gg 1$ and that it is bounded by $1$ for $r \ll -1$. It remains to give a bound on the Hilbert-Schmidt norm of the last factor on the right-hand side of Eq.~\eqref{eq:properties_of_a0_W_12}. We write $\left[ \frac{1}{z-H_0^w} \right]_{12} = \left[ \frac{1}{z-H_0} + \frac{1}{z-H_0} h^2 \omega \frac{1}{z-H_0^w} \right]_{12}$ and estimate
\begin{align}
&\left\Vert (-ih \nabla_j) W \left[ \frac{1}{z-H_0^w} \right]_{12} \right\Vert_2 \lesssim \left\Vert W \right\Vert_{H^1(\mathbb{R}^3)} \Bigg\lbrace h^{-3/2} \left\Vert p_j \hat{\Delta}_0(p) \right\Vert_{L^2(\mathbb{R}^3)} \frac{1}{1+|r|} \label{eq:properties_of_a0_W_14} \\
&\hspace{5cm} + h^2 \left\Vert p_j \frac{z+k(p)}{\left( z^2 - E(p)^2 \right)} \right\Vert_{L^{\infty}(\mathbb{R}^3)} \bigg\Vert W \left[ \frac{1}{z-H_0^w} \right]_{12} \bigg\Vert_2  \nonumber \\ 
&\hspace{5cm} + h^{1/2} \left\Vert W \right\Vert_{H^1(\mathbb{R}^3)} \left\Vert p_j \hat{\Delta}_0(p) \right\Vert_{L^2(\mathbb{R}^3)} \left\Vert \left[ \frac{1}{z-H_0^w} \right]_{22} \right\Vert_{\infty} \Bigg\rbrace. \nonumber
\end{align}
To estimate the second term on the right-hand side of Eq.~\eqref{eq:properties_of_a0_W_14}, we note that
\begin{equation}
\left\Vert \frac{p_j}{z-E(p)} \right\Vert_{L^{\infty}(\mathbb{R}^3)} \lesssim \begin{cases} |r|^{1/2} & \text{for } r \gg 1, \\ \left( \frac{1}{1+|r|} \right)^{1/2} & \text{for } r \ll -1. \end{cases} \label{eq:properties_of_a0_W_15}
\end{equation}
On the other hand, when we expand $\left( z - H_0^w \right)^{-1}$ another time and use the bound
\begin{equation}
\left\Vert \frac{1}{z-E(hp)} \right\Vert_{L^2(\mathbb{R}^3)} \lesssim h^{-3/2} \begin{cases} |r|^{1/4} & \text{for } r \gg 1, \\ \left( \frac{1}{1+|r|} \right)^{1/4} & \text{for } r \ll -1, \end{cases} \label{eq:properties_of_a0_W_16}
\end{equation}
it can easily be checked that
\begin{equation}
\left\Vert W \left[ \frac{1}{z-H_0^w} \right]_{12} \right\Vert_2 \lesssim h^{-3/2} \begin{cases} |r|^{3/4} & \text{for } r \gg 1, \\ \left( \frac{1}{1+|r|} \right)^{3/4} & \text{for } r \ll -1. \end{cases} \label{eq:properties_of_a0_W_17}
\end{equation}
Combining the estimates from Eq.~\eqref{eq:properties_of_a0_W_12}~-~Eq.~\eqref{eq:properties_of_a0_W_17}, we obtain that the Cauchy integral of the first term on the right-hand side of Eq.~\eqref{eq:properties_of_a0_W_11} is bounded by a constant times $h^{5/2}$.

It remains to estimate the contribution coming from the last term on the right-hand side of Eq.~\eqref{eq:properties_of_a0_W_11}. Following the above procedure, we would see that we do not have enough decay in $r$ to be able to assure the convergence of the Cauchy integral. To circumvent this problem, we use the algebraic identity $g(\beta z) = -g(-\beta z) + 1$. We have
\begin{align}
&\left\Vert \frac{1}{2 \pi i} \int_{\mathcal{C}} g(\beta z) (-ih \nabla_j) D \Delta_0 h^2 W D (z-k) h^2 W \left[ \frac{1}{z-H_0^w} \right]_{22} \text{d}z \right\Vert_2 \label{eq:properties_of_a0_W_21} \\
&\hspace{2.5cm} \leq \left\Vert \frac{1}{2 \pi i} \int_{\mathcal{C}} g(-\beta z) (-ih\nabla_j) D \Delta_0 h^2 W D (z-k) h^2 W \left[ \frac{1}{z-H_0^w} \right]_{22} \text{d}z \right\Vert_2 \nonumber \\
&\hspace{3cm} + \left\Vert \frac{1}{2 \pi i} \int_{\mathcal{C}} (-ih\nabla_j) D \Delta_0 h^2 W D (z-k) h^2 W \left[ \frac{1}{z-H_0^w} \right]_{22} \text{d}z \right\Vert_2. \nonumber
\end{align}
The first term on the right-hand side of Eq.~\eqref{eq:properties_of_a0_W_21} can be estimated as before. The integrand of the second term on the right-hand side of the same equation equals zero as we will see now. Let $\psi, \phi \in L^2(\mathbb{R}^3)$ be such that $\hat{\psi} \in \mathcal{C}_c^{\infty}(\mathbb{R}^3)$ and denote $\phi_z = [(z-H_0^w)^{-1}]_{22} \phi$. The term under consideration in the inner product with $\psi$ and $\phi$ reads
\begin{equation}
f(z) = \left( \psi, (-ih\nabla_j) D(z) \Delta_0 h^2 W D(z) (z-k) h^2 \tilde{W} \phi_z \right). \label{eq:properties_of_a0_W_25}
\end{equation}
It defines an analytic function on $\mathbb{C} \setminus \mathbb{R}$ and has to be integrated over the contour $\mathcal{C}$. By continuity, we can replace $W$ by another function $\tilde{W}$ with compact support without disturbing the argument to come. Using the fact that $\hat{\psi}$ and $\tilde{W}$ have compact support, one can easily justify the estimate
\begin{equation}
\left| f(z) \right| \lesssim \frac{1}{1+|y|} \frac{1}{1+|z|^3}, \label{eq:properties_of_a0_W_26}
\end{equation}
where $z=x+iy$ and $y \neq 0$. Let us consider the part of the integral over $\mathcal{C}$ that lies in the upper half-plane. Standard arguments from complex analysis and Eq.~\eqref{eq:properties_of_a0_W_26} show that
\begin{equation}
\int_{-\infty}^{\infty} f(x+iy) \text{d}x = \int_{-\infty}^{\infty} f(x+iy') \text{d}x \label{eq:properties_of_a0_W_27}
\end{equation}
holds for all $y,y'>0$. On the other hand,
\begin{equation}
\left| \int_{-\infty}^{\infty} f(x+iy) \text{d}x \right| \lesssim \frac{1}{1+|y|}  \label{eq:properties_of_a0_W_28}
\end{equation}
together with Eq.~\eqref{eq:properties_of_a0_W_27} implies that the absolute value of the integral $\int_{-\infty}^{\infty} f(x+iy) \text{d}x$ is smaller than any given positive number. The same argument can be done for the integral in the lower half-plane. This shows that the second term on the right-hand side of Eq.~\eqref{eq:properties_of_a0_W_21} equals zero and ends the proof of Lemma~\ref{lem:properties_of_alpha0_W}.  
\end{proof}
\begin{lemma}
Let $V \in L^2(\mathbb{R}^3)$, $\hat{V} \in L^{1}(\mathbb{R}^3) \cap H^2(\mathbb{R}^3) \cap W^{2,\infty}(\mathbb{R}^3)$ as well as $(1+x^2) W \in L^2(\mathbb{R}^3) \cap L^{\infty}(\mathbb{R}^3)$ and $\hat{W} \in L^1(\mathbb{R}^3)$. Then 
\begin{equation}
\left\Vert \hat{\alpha}_0^w(p,q) - \hat{\alpha}_0\left(h(p-q)\right) \right\Vert_{H^2(\mathbb{R}^6)} \lesssim h^{3/2} \label{eq:properties_of_alpha0_W_28}
\end{equation}
\label{lem:properties_of_alpha0_W_b}
holds.
\end{lemma}
\begin{proof}
The proof goes along the same lines as the one of Lemma~\ref{lem:properties_of_alpha0_W}. The only difference is that we now have to commute $x^2$ until it stands next to $W$, while in the previous proof we commuted $(-ih \nabla_j)$ until it stood next to $\Delta_0$. Since no additional difficulties arise, we leave the proof to the reader.
\end{proof}
\end{appendix}
\textbf{Acknowledgements.} The author is grateful to Christian Hainzl for proposing the problem investigated in this article as part of his doctoral studies. In addition to that, it is a pleasure to thank Christian Hainzl and Robert Seiringer for many valuable discussions and Tim Tzaneteas for proofreading an earlier version of this document. Financial support by the DFG through the Graduiertenkolleg 1838 and by the European Research Council (ERC) under the European Union's Horizon 2020 research and innovation programme (grant agreement No 694227) is gratefully acknowledged.

\vspace{0.5cm}

(Andreas Deuchert) Institute of Science and Technology Austria (IST Austria)\\ Am Campus 1, 3400 Klosterneuburg, Austria\\ E-mail address: \texttt{andreas.deuchert@ist.ac.at}

\end{document}